\documentclass[a4paper,UKenglish,cleveref, autoref, thm-restate]{lipics-v2021}
\title{Balancing Participation and Decentralization in Proof-of-Stake Cryptocurrencies}

\titlerunning{Participation and Decentralization in PoS Cryptocurrencies} %TODO optional, please use if title is longer than one line

\author{Aggelos Kiayias}{University of Edinburgh and IOG, UK}{aggelos.kiayias@iohk.io}{}{}

\author{Elias Koutsoupias}{University of Oxford and IOG, UK}{elias.koutsoupias@iohk.io}{}{}

\author{Francisco Marmolejo-Cossio}{Harvard University and IOG, USA}{fjmarmol@seas.harvard.edu}{}{}

\author{Aikaterini-Panagiota Stouka\footnote{Part of this work was conducted while Stouka was a research associate at the Edinburgh Blockchain Technology Lab}}{Nethermind, UK}{aikaterini-panagiota@nethermind.io}{}{}

\authorrunning{Kiayias et al.} %TODO mandatory. First: Use abbreviated first/middle names. Second (only in severe cases): Use first author plus 'et al.'

\Copyright{} %TODO mandatory, please use full first names. LIPIcs license is "CC-BY";  http://creativecommons.org/licenses/by/3.0/


 %TODO mandatory; please add comma-separated list of keywords

\hideLIPIcs

\category{} %optional, e.g. invited paper

\relatedversion{} %optional, e.g. full version hosted on arXiv, HAL, or other respository/website
\nolinenumbers

\usepackage{algorithm}
\usepackage{algpseudocode}

\newcommand{\vect}[1]{\boldsymbol{\mathbf{#1}}}

\newcommand\cA{\mathcal{A}}

\newcommand\cG{\mathcal{G}}

\newcommand\cX{\mathcal{X}}
\newcommand\cD{\mathcal{D}}

\EventEditors{}
\EventNoEds{0}
\EventLongTitle{}
\EventShortTitle{}
\EventAcronym{}
\EventYear{}
\EventDate{}
\EventLocation{}
\EventLogo{}
\SeriesVolume{}
\ArticleNo{0}
\begin{document}

\maketitle

\begin{abstract}
Proof-of-stake blockchain protocols have emerged as a compelling paradigm for organizing distributed ledger systems. In proof-of-stake (PoS), a subset of stakeholders participate in validating a growing ledger of transactions. For the safety and liveness of the underlying system, it is desirable for the set of validators to include multiple independent entities as well as represent a non-negligible  percentage of the total stake issued. In this paper, we study a secondary form of participation in the transaction validation process, which takes the form of stake delegation, whereby an agent delegates their stake to an active validator who acts as a stake pool operator. We study reward sharing schemes that reward agents as a function of their collective actions regarding stake pool operation and delegation. Such payment schemes serve as a mechanism to incentivize participation in the validation process while maintaining decentralization. We observe natural trade-offs between these objectives and the total expenditure required to run the relevant reward schemes. Ultimately, we provide a family of reward schemes which can strike different balances between these competing objectives at equilibrium in a Bayesian game theoretic framework.
\end{abstract}

\section{Introduction}

Proof-of-stake (PoS) blockchain protocols have emerged as a compelling paradigm for organizing distributed ledger systems. Unlike Proof-of-work (PoW), where computational resources are expended for the opportunity to append transactions to a growing ledger, PoS protocols designate the potential to update the ledger proportionally to the stake one has within the system. Common to both approaches is the fact that larger and more varied participation in the transaction validation process provides the system with increased security and resilience to faults. 

Although participating as a validator in a PoS protocol is computationally less intensive than doing so in a PoW protocol, it still demands some effort, e.g., that the validator is consistently online and maintains dedicated hardware and software, thus it is still not the case that every agent in the system decides to, or is even able to, do so. Given this,  a compelling intermediate form of participation in the transaction validation process is stake delegation. In PoS systems with stake delegation (sometimes referred to as ``liquid staking''), validators can be considered stake pool operators (SPOs), who activate pools controlling their own as well as delegated stake of others. Agents who prefer not to participate as validators have the opportunity to delegate their stake to active pools and gain rewards. In this paradigm, pools are chosen to update the ledger proportional to the combination of their ``pledged stake'' (i.e., stake they contribute)  and externally delegated stake (stake contributed to them by others); in this way, delegation can be seen as a vetting of how frequently operators should be selected to participate. Furthermore, delegation is not borne out of good will alone, since the system provides additional payments to all agents as a function of the profile of pool operators and delegators in the system. The space of payment mechanisms provides for an interesting problem in balancing three objectives: increasing participation in the validation protocol of the system (via delegation), maintaining a decentralized validation creation process (in spite of added delegation), and balancing the budget of rewards to be given to operators and delegators. 

\subsection{Related Work and Motivation}

Our work is most related to that of Brunjes et al.~\cite{brunjes2020reward} which introduces a reward sharing scheme for stake pools as a mechanism to incentivize decentralization --- a key objective shared with our work. The reward sharing scheme of their paper has been operational on the Cardano mainnet since July 2020.\footnote{\url{https://roadmap.cardano.org/en/shelley/}} In this reward sharing scheme,  the system is modulated by two system parameters, $k$, an integer representing the number of pools which are formed at equilibrium and $\alpha$, a bonus given to SPOs that distinguishes between pledged and delegated stake. This parametric formulation has the added benefit of ensuring $k$ pools of equal size being formed and preventing a single entity with low stake from controlling the majority of pools. Continuing with this line of work, Ovezik and Kiayias \cite{ovezik2022decentralization} analyze the Nash dynamics of the Cardano reward sharing scheme and the decentralization that it offers through metrics similar to those we employ to measure decentralization. In more detail, they use a variation of the Nakamoto coefficient \cite{lin2021measuring} that takes into account not only the number of pools in the system, but also the overall composition of stake of the operators who run the pools. This metric can be loosely interpreted as a measure of the composition of ``skin in the game" that SPOs have, by looking at the overall pledged stake from SPOs that may have enough cumulative stake (pledged and delegated) to perform an attack on the system. Multiple subsequent papers have proposed other metrics for decentralization of blockchain protocols (with applicability beyond PoS consensus), including \cite{azouvi2021levels},\cite{lin2021measuring},\cite{LEE2021278},\cite{karakostas2022sok}, \cite{gencer2018decentralization}, \cite{decentrfruitchain}, \cite{Leonardos2019OceanicGC}, \cite{gersbach2022staking},
\cite{ovezik2022decentralization}. 

Both \cite{brunjes2020reward} and \cite{ovezik2022decentralization} use in their analysis a framework for incentives called \textit{non-myopic utility} that tries to predict how delegators will choose a pool when the system stabilizes at equilibrium. This analysis is essential because a key component of their reward mechanism is the \textit{margin} of rewards an SPO keeps for themselves before further sharing rewards with delegators. Motivated by the above, we present a variation of the reward schemes of \cite{brunjes2020reward} in which the margin of the operators is implicitly set by the system. Most importantly, we study trade-offs between three competing objectives for the system: decentralization, overall participation, and the expenditure of the reward sharing scheme used. Furthermore, we study this performance in the presence of users who are only willing to delegate their stake if the reward they earn is lower bounded by an amount $\epsilon$ (i.e., users who may be ``lazy'', or who may have external sources of earnings for their stake). 
 
\paragraph*{Liquid Staking Protocols on Ethereum}
The framework of reward sharing schemes that we present is general enough to encompass key features of liquid staking protocols (LSPs) which are increasingly used in the Ethereum blockchain after its transition to PoS consensus \cite{gogol2024sok}. At a high level, LSPs allow users to ``stake'' their cryptocurrency (such as ETH) to be used for validation even when their cryptocurrency held is below the 32 ETH threshold required to be a validator. Upon staking their assets, users receive a liquid token in return, representing the staked assets. These liquid tokens can be used in various decentralized finance (DeFi) applications, providing liquidity and earning additional rewards while the original stake continues to generate staking rewards when used to facilitate validation. Validators for LSPs are equivalent to SPOs in our model, and individuals who mint LSPs are similar to delegators in our model. Rewards are inherently generated by the Ethereum validation process and shared according to the specification of the corresponding LSP.  

Currently on Ethereum, more than 25\% of all ETH in circulation is staked to be used for PoS validation, of which more than 50\% is attributed to 5 validators participating in LSPs (amounting overall to approximately 50 billion USD as of May 2024)\footnote{\url{https://defillama.com/protocols/liquid\%20staking/Ethereum}}. Of all staking in LSPs, the majority is deposited to the Lido LSP (approximately 29\% of the total ETH in circulation), which facilitates staking ETH to a permissioned set of validators designated by the Lido DAO \cite{lido2021}. However, most relevant to our work is the permissionless validator setting, exemplified by Rocket Pool \cite{rocketpool2021} and ether.fi \cite{etherfi2023} in which any user can become a validator as long as they have enough ETH pledged as collateral, according to the specifications of the underlying LSP. Note that in these systems, providing collateral is deemed essential for aligning incentives of validators, for without locking collateral, they can mount attacks on the LSP by shorting liquid stake tokens with nothing to lose. In our work, SPOs choose how much stake to pledge to the pool they operate, and this quantity plays the same essential role as locked collateral in LSPs, forcing SPOs to have ``skin in the game" in terms of the consensus validation process.

\subsection{Overview}

We consider a setting where a finite number of agents owns a publicly known amount of stake in a decentralized system. Agents are at a high level given three options:
\begin{itemize}
    \item They can create a stake pool, whereby they can be delegated stake from other players. Such agents are called pool operators. To be a pool operator, the agent must pledge whatever stake they own and, in addition, incur a private pool operating cost of $c> 0$.
    \item They can delegate their stake to pools that are in operation. Such agents are called delegators.
    \item They can decide to abstain from participating in the protocol and remain idle, earning baseline utility $\epsilon> 0$. 
\end{itemize}

%It is important to note that this setting assumes that each unit of stake in the system can be attributed to a single owner (this is inherent in the fact that our model permits each agent to take only one of the 3 high-level actions above). In other words,
Note we do not model the scenario where agents can create multiple identities (i.e. perform sybil attacks), or where they can pool resources outside of the system and coordinate as what seems to be a single agent in the system. 
We stress that the scope of this paper is to show that there are still important trade-offs (Decentralization, Participation and System Expenditure) that system designers need to consider in the setting where agents are identified via their wallet.
Broadening the model to permit a richer set of agent behaviours, such as Sybil attacks, is an important future area of research. 

%It is important to note that this setting assumes that each unit of stake in the system can be attributed to a single owner (this is inherent in the fact that our model permits each agent to take only one of the 3 high-level actions above). In other words, we do not model the scenario where agents can create multiple identities (i.e. perform sybil attacks), or where they can pool resources outside of the system and coordinate as what seems to be a single agent in the system. We stress that the scope of this paper is to show that there are important trade-offs (Decentralization, Participation and System Expenditure) that system designers need to consider in the setting where agents are identified in a system (for example via KYC). Indeed, we believe that broadening the model to permit this agent behaviour is an important future area of research.

%\fm{Add something about lack of sybil-attacks, as a part of the model, but we're not including that in this paper, and also how we don't consider agents joining offline. We have a bunch of agents, we know their stake and that's it. It's a model where we have identified the agents.
% do we want to cite? 
%Yujin Kwon, Jian Liu, Minjeong Kim, Dawn Song, Yongdae Kim: Impossibility of Full Decentralization in Permissionless Blockchains. AFT 2019: 110-123
%}

\paragraph*{Participation}  We are interested in systems that encourage increased participation in the overall validation process. To prevent agents from abstaining from the protocol (and hence participating), they must at least be able to delegate in such a way as to earn more than $\epsilon$, their baseline utility for remaining idle as far as staking is concerned. 

\paragraph*{Rewards and Incentives} The aforementioned structure alone does not provide incentives to drive agents' actions. To create such incentives, we consider reward schemes whereby pool operators and delegators are compensated as a function of which pools are active and whom delegators choose to delegate to. As we will see in the following section, this creates a well-defined family of one-shot games that are played between all agents in the system, and we study the equilibria that result as a function of the reward scheme implemented.

\paragraph*{Informal Design Objectives} Our main objective is to create reward schemes that optimise for three distinct objectives: 
\begin{itemize}
    \item Increasing participation in the system.
    \item Increasing Decentralization, i.e. preventing stake from overly accumulating (via delegation) in the hands of few pool operators.
    \item Minimizing the budget necessary to achieve the above.
\end{itemize}

\subsection{Roadmap of our Results}

We consider the setting in which stakeholders of a PoS blockchain can either operate pools (receive delegation), delegate their stake, or abstain from the protocol, where each of these actions provides a certain reward from the system. Section \ref{sec:delegation-game} begins by introducing the notion of a delegation game, which is a general framework for encapsulating strategic considerations between stakeholders in this setting. At the end of Section \ref{sec:delegation-game}, we introduce the notion of a uniform reward delegation game, which is a refinement of general delegation games by which all delegators in the system (roughly) earn a uniform reward per unit of stake that they delegate. Within the class of uniform delegation games we further hone our focus on proper delegation games which we define in such a way to exemplify relevant characteristics of existing reward sharing schemes deployed in practice. In Section \ref{sec:PNE-in-proper-delegation-game} we provide sufficient conditions for pure Nash equilibria in proper delegation games. Section \ref{sec:bayesian-setting} introduces a Bayesian framework to proper delegation games and explores novel solution concepts intricately tied to ex post pure Nash equilibria. In Section \ref{sec:objectives} we introduce the main metrics by which we compare the equilibria of the Bayesian proper delegation game: participation, decentralization and system expenditure. Section \ref{sec:computational-methods-results} provides details on the computational methods used to evaluate the performance of payment schemes in proper delegation games at equilibrium, along with experimental results. Finally, Section \ref{sec:conclusion} provides a conceptual overview of the results obtained and provides future directions of work.  

\section{The Delegation Game}
\label{sec:delegation-game}

We now formalize the general family of games which govern agent decisions regarding whether to create a pool or delegate their stake. We consider $n > 0$ players, each with a publicly known stake, $s_i > 0$. Additionally, we assume that any agent who chooses to operate a pool and participate actively incurs a fixed cost of $c_i > 0$. Finally, we assume that each player has a fixed utility for non participation in delegation, which we denote by $\epsilon_i > 0$. Such a utility can encompass the fact that an agent may find participating in stake delegation prohibitively complicated, or that they prefer using their stake in other ways (such as other governance or DeFi protocols, for example).   
 
\paragraph*{Player Strategies} For each player, $i \in [n]$, let $\mathcal{D}_i$ denote the set of functions $d_i: [n]\setminus\{i\} \rightarrow \mathbb{R}^{+}$ such that $\sum_{j \in [n]\setminus\{i\}} d_i(j) = s_i$. The action space of the $i$-th player corresponds to the set $\mathcal{A}_i = \{a_I\} \cup \{a_{SPO}\} \cup \mathcal{D}_i$. We further denote the space of all joint strategy profiles by $\cA = \prod_i \cA_i$. A joint strategy profile of the game is a vector $\vect{p} = (p_i)_{i=1}^n \in \cA$, where $p_i \in \cA_i$ denotes the action taken by the $i$-th agent. Furthermore, for a fixed agent $i \in [n]$, we let $\cA_{-i}$ denote the action space of all players other than $i$, such that $\vect{p}_{-i} \in \cA_{-i}$ denotes a specific collection of strategies for all players in $[n] \setminus \{i\}$, and $\vect{p} = (p_i,\vect{p}_{-i}) \in \cA$ denotes a strategy profile that makes specific reference to the action $p_i \in \cA_i$ played by the $i$-th player. There are 3 relevant cases for the values $p_i$ can take and hence the actions that the $i$-th player can take:
\begin{itemize}
    \item $p_i = a_{I}$ represents non-participation in delegation for the $i$-th agent. We say that the agent is {\it idle}.
    \item $p_i = a_{SPO}$ occurs when the $i$-th player chooses to operate their pool. To do so, they pledge their stake, $s_i$, to the pool and incur a pool operation cost of $c_i$. We say the agent is a {\it stake pool operator (SPO)}. 
    \item $p_i = d_i \in \cD_i$ occurs when the $i$-th player chooses to delegate their stake, $s_i$, to different pools operated by other agents. We call $d_i$ the player's {\it delegation profile}. For each $j \in [n]\setminus \{i\}$, the player $i$ delegates $d_i(j)$ stake to a pool operated by the $j$-th agent. We say that the agent is a {\it delegator}. 
\end{itemize}
\begin{definition}[Active-Inactive Pool]
A pool $j$ will be called active in the joint strategy profile $\vect{p} \in \cA$ if $p_j= a_{SPO}$. That is, if player $j$ has pledged their stake, $s_j$, to operate their pool. If this is not the case, we say that the pool $j$ is inactive.
\end{definition}

\paragraph*{Rewards}
For each agent, $i \in [n]$, we let $R_i: \cA \rightarrow \mathbb{R}$ be their delegation game reward function. If $\vect{p}\in \cA$ is a joint strategy profile of all agents, $R_i(\vect{p})$ denotes the reward obtained by the $i$-th agent. We impose the following constraints on $R_i$:
\begin{itemize}
    \item If $p_i = a_{I}$, then $R_i(\vect{p}) = \epsilon_i$. 
    \item If $p_i = d_i \in \cD_i$, then the reward, $R_i(d_i,\vect{p}_{-i})$ can be further decomposed as the sum of $n-1$ delegation reward functions: $R_i(d_i,\vect{p}_{-i}) = \sum_{j \in [n]\setminus\{i\}} R_{i,j}(d_i(j),\vect{p}_{-i})$ which satisfy two constraints:
    \begin{itemize}
        \item $R_{i,j}(0,\vect{p}_{-i}) = 0$ for all $\vect{p}_{-i} \in \cA_{-i}$. That is, no rewards can be earned by abstaining from delegating to a given pool.
        \item If  pool $j$ is not active (that is, $p_j \neq a_{SPO}$), then $R_{i,j}(d_i(j),\vect{p}_{-i}) = 0$. More succinctly, if a player delegates stake to an inactive pool, they receive no reward.
    \end{itemize}
\end{itemize}

\paragraph*{Utilities}
For each $i \in [n]$, we let $u_i: \cA \rightarrow \mathbb{R}$, denote the $i$-th player's utility, given by $u_i(\vect{p}) \in \mathbb{R}$ for a joint strategy $\vect{p} \in \cA$. 
%Furthermore, we assume that each agent has the same baseline utility for remaining idle, which is denoted by $U_I > 0$.
In our setting, we define utilities in terms of the aforementioned reward function:
\begin{equation}
\label{eq:utilities-delegation-game}
  u_{i}(\vect{p})) =
    \begin{cases}
      %U_I & \text{if $p_i = x_{I}$}\\
      \epsilon_i & \text{if $p_i = a_{I}$}\\
      R_i(\vect{p})) - c_i & \text{if $p_i = a_{SPO}$}\\
      R_i(\vect{p}))  & \text{if $p_i \in \cD_i$}      
    \end{cases}       
\end{equation}

\begin{definition}[The Delegation Game]
\label{delegationgame}
Suppose that we have $n$ agents with publicly known stakes denoted by $\vect{s} = (s_i)_{i=1}^n$, privately known pool operation costs $\vect{c} = (c_i)_{i=1}^n$ and privately known idle utilities $\vect{\epsilon} = (\epsilon_i)_{i=1}^n$. In addition, suppose that $\vect{R} = (R_i)_{i=1}^n$ is a family of reward functions $R_i: \cA \rightarrow \mathbb{R}^{+}$. 
%and that $U_I > 0$ is their idle utility. 
We let $\mathcal{G}(\vect{R},(\vect{s},\vect{c},\vect{\epsilon}))$ be the corresponding game with induced utilities $\vect{u} = (u_i)_{i=1}^n$ from above. This game is called the ``Delegation Game'' for $\vect{s},\vect{c}$,$\vect{\epsilon}$, and $\vect{R}$.    
\end{definition}

\subsection{Games with Uniform Delegation Rewards}

Given the large class of delegation games described above, we focus on a natural class of delegation games similar to what is used on the Cardano blockchain \cite{brunjes2020reward}. Cardano rewards have the following relevant high-level characteristics: 

\begin{enumerate}
\item Each pool $j$ receives a total amount of rewards according to a {\it pool reward function} that takes as input the stake of the pool operator and the stake delegated to the pool.
%\kate{02.07}{ the second input is the total stake of the pool, but not sure if we should change it here}
\item The pool operator may keep an amount of the pool rewards. They do so by picking a margin of pool rewards to keep. 
\item The remaining pool rewards (called \textit{Pool Member Rewards}) are proportionally shared. 
%\kate{31.08}{maybe we should do it ``are shared proportionally to their delegated stake" ..}
%I think it's ok since it might be confusing with the distinction between pledging and delegating. This is expanded upon later anyways. 
amongst the pool operator and delegates to the pool. 
\end{enumerate}

The subclass of delegation games we study in this paper will incorporate similar pool reward functions, hence to proceed, we define the following important terms that result from a joint strategy profile $\vect{p} \in \cA$:
 \begin{itemize}
 \item $\beta_j$: the external stake delegated to pool $j$ under $\vect{p}$. This is given by $\beta_j = \sum_{i:p_i \in \cD_i} d_i(j)$.
 \item $\lambda_j$: the operator pledge of pool $j$. This is given by $\lambda_j=s_j$, when $p_j=a_{PO}$; otherwise it is $\lambda_j=0$.
 \item $\sigma_j$: the total stake of a pool $j$. This is given by $\lambda_j + \beta_j$.
  \end{itemize}

\begin{definition}[Pool Reward Function]
A pool reward function is given by $\rho:(\mathbb{R}^{+})^2 \rightarrow \mathbb{R}^{+}$ that takes as input the pledged stake of its pool leader, $\lambda_j$, and the external stake delegated to the pool, $\beta_j$ and outputs the rewards that correspond to pool $j$, given by $\rho(\lambda_j,\beta_j)$. 
\end{definition}

As detailed in \cite{brunjes2020reward}, the Cardano pool reward function has the further property that rewards are capped (so that pools stop earning surplus rewards once they reach a certain size), and the rewards themselves can be decomposed into a specific algebraic form which we call separable: 

\begin{definition}[Capped Separable Pool Reward Function]
Let $\tau > 0$ and $a,b: \mathbb{R}^+ \rightarrow \mathbb{R}^+$ and define $\rho:(\mathbb{R}^{+})^2 \rightarrow \mathbb{R}^{+}$ as follows:
$$
\rho(\lambda,\beta) = a(\lambda') + b(\lambda')\beta',
$$  
where $\lambda' = \min\{\tau,\lambda\}$ and $\beta' = \min\{\tau - \lambda', \beta\}$. We say that $\rho:(\mathbb{R}^{+})^2 \rightarrow \mathbb{R}^{+}$ is a capped pool reward function with a cap given by $\tau$. In addition, we say that $\rho$ is separable into $a$ and $b$, where $a$ is the pool's {\it pledge reward component} and $b$ is the pool's {\it external delegation reward component}.
\end{definition}

Upon close inspection, Delegation games, as per Definition~\ref{delegationgame}, already exemplify an important point of departmure from Cardano reward sharing schemes. Namely, our setting has a simpler action space for agents amounting to mostly the high-level choice of: being an SPO, being a delegator, and being idle. In Cardano, rewards have a more complicated action space whereby beyond the choice to become an SPO, agents can also pick the margin of rewards they wish to keep as SPOs. In \cite{brunjes2020reward}, the authors study the parametric family of pool reward functions used in Cardano to show that when players are non-myopic, one can modulate the number of pools, $k$, which are formed at equilibrium. An important characteristic of these equilibria though is the fact that pool operators choose a margin such that delegators are indifferent amongst the $k$ active pools in terms of the delegation reward they obtain from them (i.e. the proportional rewards after margins are taken by pool operators). Rather than letting agents reach such an outcome at equilibrium, we study delegation games with the very property that delegators earn the same per-unit reward mostly irrespective of the pool to which they delegate. In order to do so, we introduce the notion of delegator rewards:

\begin{definition}
A delegation reward function is given by $r:\cA \times (\mathbb{R}^{+})^n \rightarrow \mathbb{R}^{+}$ which takes as input the publicly known joint strategy $\vect{p} = (p_i)_{i=1}^n$ and stake distribution $\vect{s} = (s_i)_{i=1}^n$
%, and participation costs $\vect{c} = (\vect{c}^p,\vect{c}^d)$ such that $c^p_i > c^d_i >0$ for $i \in [n]$
to output a \textit{fixed} reward per unit of delegated stake given by %$r(\vect{p},\vect{s},\vect{c})$.
$r(\vect{p},\vect{s})$.
\end{definition}

We will shortly precisely define delegation games with uniform delegation rewards, but at a high level these games have reward functions that automatically enforce the fact that for a given strategy profile, delegators will receive $r(\vect{p},\vect{s})$ rewards per unit of delegation. Continuing with the comparison with Cardano, at equilibrium, it is not the case that all pools have equal per-unit delegation rewards, but rather the $k$ pools which offer the best per-unit delegation rewards to delegators %(\kate{23.07}{the following it is not always the case because it depends on a parameter which pools are in the equilibrium but I do not know if it worth making it more complicated} \fm{As long as we can say that all pools don't have the same per-unit rewards we're fine.}
%\kate{31.08}{it would be more accurate to replace the following sentence with ''which are, in turn, those pools with the most profitable combination of pledge and cost''} 
which are, in turn, those pools with the most profitable combination of pledge and cost). It can very well be the case that a suboptimal pool remain in operation, albeit offering lower per-unit rewards to potential delegators. In this spirit, we define the notion of pool feasibility, which serves as a way to determine which pools are suboptimal. Suboptimality will mean that the cumulative earnings of all agents involved in a pool (including the SPO) is less than what they would earn as delegators according to the delegation reward function $r$.  

\begin{definition}[Pool feasibility]
For a given joint strategy profile $\vect{p}$, we call the $i$-th pool {\em feasible} if $p_i=a_{SPO}$ and
$\rho(\lambda_i, \beta_i) \geq \sigma_i r(\vect{p},\vect{s})$.
\end{definition}

Now we have everything in hand to define the notion of a delegation game with uniform delegate rewards. We specify the rewards that each agent earns in the game.

\begin{definition}[Uniform Delegation Agent Rewards]
\label{uniform}
Suppose that we have $n$ agents with stake distribution $\vect{s}$, participation costs $\vect{c}$, and idle utilities $\vect{\epsilon}$. Furthermore, suppose that $\vect{p} \in \cA$ is a joint strategy profile such that $p_i = d_i \in \cD_i$. If we let $r = r(\vect{p},\vect{s})$, then the components of the reward function for the $i$-th agent are: 
\begin{equation}
\label{eq:constant-delegation-delegate-rewards}
  R_{i,j}(d_i(j),\vect{p}_{-i}) =
    \begin{cases}
      r \cdot d_i(j) & \text{if pool $j$ is active and feasible}\\
      \frac{d_i(j)}{\sigma_j} \cdot \rho(\lambda_j,\beta_j) & \text{if pool $j$ is active and not feasible}\\
      0 & \text{if pool $j$ is not active}
    \end{cases}       
\end{equation}
With this in hand, we can fully define the reward function for the $i$-th agent under arbitrary actions as follows:
\begin{equation}
\label{eq:constant-delegation-pool-rewards}
  R_{i}(\vect{p}) =
    \begin{cases}
      \epsilon_i & \text{if $p_i = a_I$}\\
      \rho(\lambda_i,\beta_i) - r \cdot \beta_i & \text{if $p_i = a_{SPO}$ and pool $i$ is feasible}\\
      \frac{\lambda_i}{\sigma_i} \cdot \rho(\lambda_i,\beta_i) & \text{if $p_i = a_{SPO}$ and pool $i$ not feasible}\\
      \sum_{j \in [n]\setminus\{i\}} R_{i,j}(d_i(j),\vect{p}_{-i}) & \text{if $p_i = d_i\in \cD_i$}
    \end{cases}       
\end{equation}
If a delegation game $\cG$ has uniform delegation rewards, we say it is a {\it uniform delegation reward game}. 
\end{definition}

\subsubsection{Narrowing Down Delegation Rewards}

The final component we need to specify in order to delve into delegation game equilibria is the delegation reward function that we use. In \cite{brunjes2020reward} the authors show that at equilibrium, delegator rewards are essentially specified by the most competitive agent who misses out on becoming an SPO. Essentially, if one ranks pools according to potential profits at saturation, then there are equilibria where the top $k$ pools are active and have margins such that the cut of rewards which go to delegators for each of these pools equals the potential profit of the potential $(k+1)$-th pool. We recall that $k$ is a parameter of the reward sharing scheme that is intended to modulate the number of pools in the system. Moreover, this phenomenon intuitively makes sense, for the top $k$ agents are essentially as aggressive as possible in setting their margins without falling behind the $(k+1)$-th pool in desirability to potential delegators.  

In this vein, we focus on a delegation reward function that is specified according to the ``most competitive'' delegator, with the property that once such a delegator is identified, all less competitive delegators will be content with their choice in delegating. In order to proceed, we introduce new notation and terminology. 

\begin{definition}
For a given pool reward function, $\rho$, we let $\alpha: (\mathbb{R}^+)^2 \rightarrow \mathbb{R}^+$ be such that:
$$
\alpha(s,c) = \frac{\rho(s,0) - c}{s}.
$$ 
In other words, $\alpha(s,c)$ is the rewards per unit of stake that an individual with stake $s$ and pool operation cost $c$ obtains for opening a pool without external delegation (a solo pool). We call $\alpha(s,c)$ the threat of deviation of a delegator with stake $s$ and pool operation cost $c$. %\kate{03.07}{$c^p_i?$} 
\end{definition}

For a given joint strategy profile, $\vect{p}$, we would ideally want to set delegation rewards to be the maximum threat of deviation among delegators, as this would achieve our desired goal of ensuring that all delegators do not have an incentive to deviate from delegating into becoming solo pools. The problem with this, though, is that the threat of deviation fundamentally depends on each delegate's private cost of pool operation. For this reason, we suppose that there is public knowledge regarding bounds on pool operation costs, so that $0 \leq c_{min} \leq c_i \leq c_{max}$ for any $i \in [n]$. With this in hand, we define the max-delegate rewards:

\begin{definition}[Max-delegate $r$]
For a given pool reward function, $\rho$, we let $r_M:\cA \times (\mathbb{R}^{+})^n \rightarrow \mathbb{R}^{+}$ be such that:
$$
r_M(\vect{p},\vect{s}) = \max_{i: p_i \in \cD_i} \alpha(s_i,c_{min})
$$
If $\{i \in [n] \mid p_i \in \cD_i\} = \emptyset$, then we let $r_M(\vect{p},\vect{s}) = 0$
\end{definition}

Since $\alpha$ is a decreasing function in $c$, it follows that for a given joint strategy profile, $\vect{p}$, every delegator will not increase their utility by becoming a solo pool operators under $r_M$. In what follows, we will consider pool reward functions $\rho$ with the natural property that $\alpha$ is monotonically increasing in $s$ as well (i.e. per-unit solo pool rewards are increasing in SPO pledge). In this case, we can express the max-delegate reward function in a more simple and useful fashion by making use of the following:

\begin{definition}
Suppose that $\cG$ is a delegation game and that we consider a joint strategy profile $\vect{p}$. We let $s^* = \max_{i: p_i \in \cD_i} s_i$ and call this quantity the pivotal delegation stake of $\vect{p}$. If $p_i \in \cD_i$ and $s_i = s^*$, then we also say that the player is a pivotal delegate in $\vect{p}$. 
\end{definition}

If the pool reward function, $\rho$, is such that $\alpha$ increases monotonically in $s$, then it follows that $r_M(\vect{p},\vect{s}) = \alpha(s^*,c_{min})$. 
%We call such pool reward functions {\it proper}.  

\paragraph*{Putting Everything Together}
Going forward, we focus on uniform delegation games with max-delegate rewards such that per-unit solo pool delegation ($\alpha$) is monotonically increasing in pledge. We give this class of games a specific name as the main focus of this paper:

\begin{definition}[Proper delegation game]
\label{def:proper-delegation-game}
Suppose that $\cG$ is a uniform delegation game such that the following hold:
\begin{itemize}
\item The pool reward function, $\rho$ is such that per-unit solo SPO rewards, $\alpha(s,c)$, are monotonically increasing for $s \in [0,s_{max}]$, where $s_{max} = \max\{s_i\}$.
\item $\rho$ is capped and separable with $s_{max} < \tau$.
\item Delegation rewards are given by $r_M$, the max-delegate reward function.
\end{itemize}
Then we say that $\rho$ is a proper reward function and that $\cG$ is a proper delegation game. When we wish to be more specific regarding a given game, we use the notation $\cG(\rho,\tau,(\vect{s},\vect{c},\vect{\epsilon}))$ to specify the reward function and cap used, as well as the attributes of all players in the game. 
\end{definition}

\section{Equilibria in Proper Delegation Games}
\label{sec:PNE-in-proper-delegation-game}
%\section{Equilibria in the Bayesian Setting}
%\label{sec:equilibria-in-bayesian}

In the previous section, we rigorously defined the class of proper delegation games which we focus on in this paper. This section provides sufficient conditions for a joint strategy profile to be a pure Nash equilibrium.

\subsection{Sufficient Conditions for Pure Nash Equilibrium (PNE)}

We use the shorthand $r = r_M(\vect{p},\vect{s})\in \mathbb{R}^+$ to refer to the per-unit reward for delegating to a feasible pool and we begin by providing multiple structural results related to the best responses agents may have in a proper delegation game.  

\subsubsection{Structural Results regarding Best Responses}

We begin by showing that infeasible pools are always suboptimal for both SPOs and delegators. 

\begin{lemma}[Feasible pool structural lemma]
\label{lemma:feasible-pool-structural-lemma}
Suppose that $\cG(\rho,\tau,(\vect{s},\vect{c},\vect{\epsilon}))$ is a proper delegation game and all agents are playing the joint strategy profile $\vect{p}$ where the $i$-th player is an SPO ($p_i = a_{SPO}$) for an infeasible pool with pledge $\lambda_i = s_i$ and external delegation $\beta_i \geq 0$. The following hold: 
\begin{itemize}
    \item Delegators to the infeasible pool obtain strictly more utility by delegating to feasible pools.
    \item The SPO earns strictly more utility by using their pledge to delegate to feasible pools.
    %\kate{23.07}{Here I think that we need to add that if the SPO becomes a delegator $r$ may change to $r'$ (as the set of the delegators changes) but $r' \geq r$.}
    %This case is now subsumed in the proof
\end{itemize}
\end{lemma}

\begin{proof}
The infeasibility of the pool implies that $\rho(\lambda_i,\beta_i) < r \sigma_i = r \cdot (\lambda_i + \beta_i)$ by definition, where we recall that $\sigma_i = \lambda_i + \beta_i$ is the total stake of the pool (including pledge and external delegation). Suppose that a delegator has delegated $x \leq \beta_i$ stake to the pool. The infeasibility of the pool also implies that said delegator's rewards amount to 
$$
\frac{x}{\sigma_i} \rho(\lambda_i,\beta_i) < \frac{x}{\sigma_i} r\sigma_i = rx.
$$
If the SPO becomes a delegator to a feasible pool, they will earn $r'x$, where $r' \geq r$ (since they could change the per-unit delegation if they are a pivotal delegate). This concludes the proof of the first statement.

As for the second statement, the infeasibility of the pool means that the SPO earns the following rewards: 
$$
\frac{\lambda_i}{\sigma_i} \rho(\lambda_i,\beta_i) < \frac{\lambda_i}{\sigma_i} r\sigma_i = r\lambda_i.
$$
The SPO stands to earn $r \lambda_i$ rewards if they instead delegate their stake used as a pool pledge to a feasible pool, thus proving the second statement. 
\end{proof}

We now prove lemmas regarding the best responses of agents who are idle, delegators, and SPOs, respectively.

\begin{lemma}[Idle best response]
    Consider a proper delegation game $\cG(\rho,\tau,(\vect{s},\vect{c},\vect{\epsilon}))$ and a joint strategy profile $\vect{p} = (a_I,\vect{p}_{-i})$ such that $i$-th player is idle. The $i$-th player's best response to $\vect{p}$ is either remaining idle or delegating to a feasible pool.   
\end{lemma}

\begin{proof}
    This is a straightforward extension of definitions. We simply show that the deviation where the $i$-th player becomes an SPO is weakly dominated by the deviation where the $i$-th agent becomes a delegator. The deviation where the agent becomes an SPO is unilateral, hence the pool they create forcibly lacks external delegation. As such, their solo pool utility is given by $\alpha(s_i,c_i) \cdot s_i$. On the other hand, let $\vect{p}' = (p_i',\vect{p}_{-i})$ be the deviation where the $i$-th player delegates to feasible pools, resulting in per-unit delegation rewards $r'$. By definition, $r' \geq \alpha(s_i,c_{min})$, as it is the maximum value of $\alpha(s_j,c_{min})$ among the agents who delegate, which includes the $i$-th agent. Since the $i$-th player delegates to feasible pools in $\vect{p}'$, it follows that their utility is given by $r' s_i$ in the deviation. We have the following strings of inequalities: 

    \begin{equation} %\label{eq1}
    \begin{split}
    \alpha(s_i,c_i) \cdot s_i & \leq  \alpha(s_i,c_{min}) \cdot s_i \\
     & \leq r' s_i  
    \end{split}
    \end{equation}
    where we have additionally made use of the fact that $\alpha$ is decreasing in its second argument. The claim follows. 
\end{proof}

\begin{lemma}
Suppose that $\cG(\rho,\tau,(\vect{s},\vect{c},\vect{\epsilon}))$ is a proper delegation game. For any joint strategy profile $\vect{p}$, delegates to feasible pools cannot benefit from deviating to becoming SPOs.
\end{lemma}
\begin{proof}
This is an easy consequence of the fact that $\alpha$ is monotonically increasing in pledge and monotonically decreasing in pool operation cost. We recall that $s^*$ is the pivotal delegate stake for $\vect{p}$. Suppose $p_i \in \cD_i$, where the $i$-th player with stake $s_i$ and pool operation cost $c_i$ delegates to a feasible pool in $\vect{p}$. Per-unit rewards for this delegate are $r = \alpha(s^*,c_{min})$ where $s_i \leq s^*$. Monotonicity gives: 
$$
\alpha(s_i,c_i) \leq \alpha(s^*,c_i) \leq \alpha(s^*,c_{min}) = r
$$
and the per-unit reward the delegate can earn from becoming a solo SPO is in fact $\alpha(s_i,c_i)$.

\end{proof}

In what follows, we consider an SPO with pledge, pool operation cost, and idle utility given by $(\lambda,c,\epsilon)$. Moreover, we continue to let $r$ be per-unit rewards for delegating to feasible pools. We call the following quantity the ``Gap'' of the given SPO:
$$
G(\lambda,c,\epsilon,r) = \max \{\epsilon + c - a(\lambda), [r - \alpha(\lambda,c_{min})]^+ \cdot \lambda + (c - c_{min})  \} > 0,
$$
where we use the notational shorthand $[x]^+ = \max\{x,0\}$. Furthermore, when the context is clear, we simply use $G$ to refer to the gap of an SPO. 
%If we recall that $r = \alpha(\lambda^*,0) = a(\lambda^*)/\lambda^*$, we can rewrite the expression as:
%$$
%G(\lambda,c^p,r) = \left[ (\alpha(\lambda^*,0) - \alpha(\lambda,0)) \right]^+ \cdot \lambda + c^p
%$$
\begin{lemma}
\label{lemma:SPO-content}
Suppose that an SPO has $s$ stake, pool operation cost $c$, and idle utility $\epsilon$. Additionally suppose that they operate a pool with pledge $\lambda = s$ and external delegation $\beta$. The SPO cannot benefit from unilaterally deviating from pool operation (by either becoming idle, becoming a delegator or opening a new pool) if and only if:  
$$
b(\lambda) \beta' - r\beta \geq G(\lambda,c,r,\epsilon) > 0
$$
\end{lemma}
%\begin{lemma}
%Suppose that delegation rewards are defined by $r = \alpha(\lambda^*,0)$ and that an SPO with stake/cost given by $(\lambda,c)$, where forcibly $\lambda \geq \lambda^*$. Furthermore, suppose the SPO has $\beta$ external delegation. The SPO is content if and only if:
%$$
%b(\lambda) \beta' - r\beta \geq G(\lambda,c,r) \geq 0
%$$
%\end{lemma}
\begin{proof}
We start by providing algebraic conditions for the SPO to prefer operating the pool to becoming idle. The utility for operating a pool is given by $u^P = a(\lambda) + b(\lambda)\beta' -r\beta - c$, whereas the utility for remaining idle is given by $u^I = \epsilon$. It is thus clear that $u^P \geq u^I$ if and only if:
$$
b(\lambda)\beta' - r\beta \geq \epsilon + c - a(\lambda)
$$

%There are two relevant deviations we consider: 1) becoming a delegator, 2) becoming an entirely new pool and shedding the external delegation. 
Now we provide algebraic conditions for an SPO to prefer operating the pool to becoming a delegator or a solo pool. To begin, we show that becoming a delegator is always a preferable deviation to shedding external delegation and becoming a solo pool. By becoming a delegator, the per-unit reward of the agent is at least $\alpha(\lambda,c_{min})$ by definition of $r_M$. If the agent becomes a solo pool operator, however, their per-unit reward is given by $\alpha(\lambda,c) \leq \alpha(\lambda,c_{min})$. With this in hand, we only consider deviations consisting of becoming a delegator going forward. In what follows we will show that an SPO prefers running their pool over becoming a delegator if and only if: 
$$
b(\lambda)\beta' - r\beta \geq [r - \alpha(\lambda,c_{min})]^+ \cdot \lambda + (c - c_{min}).
$$
Once we prove this constraint the lemma follows, as the gap is the larger value of both of these constraints on $b(\lambda)\beta' - r\beta$. 

There are two relevant cases when considering a deviating SPO depending on whether $\lambda \leq s^*$ where we recall that $s^*$ is the pivotal stake of $\vect{p}$. 

\paragraph*{Case 1: $\lambda \leq s^*$.}
The utility the SPO has from operating the pool as is is given by:
$$
u^P = a(\lambda) + b(\lambda) \beta' -r\beta - c
$$
Whereas the utility for delegating is given by:
$$
u^D = r\lambda  = \alpha(s^*,c_{min})\lambda = \left( \frac{a(s^*) - c_{min}}{s^*} \right) \lambda,
$$
where we have used the fact that $\lambda \leq s^*$ in the fact that the same $r$ is the per-unit delegation reward after deviating. The SPO prefers the status quo if and only if $u^P \geq u^D$. If we re-arrange said inequality, we obtain the desired equivalent condition: 
\begin{equation} %\label{eq1}
\begin{split}
u^P & \geq u^D \\
a(\lambda) + b(\lambda) \beta' -r\beta - c & \geq r\lambda \\
b(\lambda) \beta' -r\beta & \geq r\lambda - a(\lambda) + c \\
b(\lambda) \beta' -r\beta & \geq r\lambda - (a(\lambda) - c_{min}) + c - c_{min} \\
b(\lambda) \beta' -r\beta & \geq r\lambda - \alpha(\lambda,c_{min})\lambda + c - c_{min} \\
b(\lambda) \beta' -r\beta & \geq (r - \alpha(\lambda,c_{min})) \cdot \lambda + c - c_{min} \\
b(\lambda) \beta' -r\beta & \geq [r - \alpha(\lambda,c_{min})]^+ \cdot \lambda + (c - c_{min})\\
\end{split}
\end{equation}
In the final line we use the fact that $\lambda \leq s^*$ implies that $\alpha(\lambda,c_{min}) \leq r$ due to the definition of $r$ and the monotonicity of $\alpha$ in its first argument.

\paragraph*{Case 2: $\lambda > s^*$}
The utility the SPO obtains from operating the pool as is is given by:
$$
u^P = a(\lambda) + b(\lambda) \beta' -r\beta - c
$$
Whereas the utility for delegating is given by:
$$
u^D = r\lambda = \alpha(\lambda,c_{min})\lambda = a(\lambda) - c_{min} ,
$$
where we have used the fact that $\lambda > s^*$ in the fact that the same $r = \alpha(\lambda,c_{min})$ is the per-unit delegation reward after deviating. The SPO prefers the status quo if and only if $u^P \geq u^D$. If we re-arrange said inequality, we obtain the desired equivalent condition: 
\begin{equation} %\label{eq1}
\begin{split}
u^P & \geq u^D \\
a(\lambda) + b(\lambda) \beta' -r\beta - c & \geq a(\lambda) - c_{min} \\
b(\lambda) \beta' -r\beta & \geq  c - c_{min} \\
b(\lambda) \beta' -r\beta & \geq  [r - \alpha(\lambda,c_{min})]^+ \cdot \lambda + (c - c_{min})
\end{split}
\end{equation}
In the final line, we used the fact that $\lambda > s^*$ implies that $\alpha(\lambda,c_{min}) = r$.

\end{proof}

\subsubsection{Pool Deficit and Capacity}

With the previous lemma in hand, we precisely characterize at what values of external delegation an SPO prefers to maintain their pool (rather than becoming a delegator or abandoning their given external delegation for a solo pool). To do so, we define the following important quantities:
\begin{definition}[Pool Deficit/Capacity]
    Consider a proper pool delegation game given by $\cG(\rho,\tau,(\vect{s},\vect{c},\vect{\epsilon}))$ where the pool reward function is given by $\rho(\lambda,\beta) = a(\lambda) + b(\lambda)\beta'$. Let $\vect{p}$ be a joint strategy profile of $\cG$ such that per unit delegation reward is given by $r$ and such that the $i$-th player is an SPO with pledge $\lambda_i < \tau$ and pool operation cost $c_i$.     
    %Suppose that a pool has pledge $\lambda < \tau$, pool operation cost $c$, and that the per-unit delegation reward is given by $r$ under a separable and capped reward function $\rho(\lambda,\beta) = a(\lambda) + b(\lambda)\beta'$ with pool cap $\tau$. 
    We let $\beta^-_i = \beta^-(\lambda_i,c_i,\epsilon_i,r)$ and $\beta^+_i = \beta^+(\lambda_i,c_i,\epsilon_i,r)$ denote the deficit and capacity, respectively, of the pool run by the $i$-th player as an SPO. The quantities are defined as follows:

\begin{equation}
\beta^-(\lambda_i,c_i,\epsilon_i,r)=
    \begin{cases}
        \frac{G(\lambda_i,c_i,\epsilon_i,r)}{b(\lambda_i) - r} & \text{if } (b(\lambda_i) - r)(\tau - \lambda_i) \geq G(\lambda_i,c_i,\epsilon_i,r)\\
        \infty & \text{otherwise }
    \end{cases}
\end{equation}

\begin{equation}
\beta^+(\lambda_i,c_i,\epsilon_i,r)=
    \begin{cases}
        \frac{b(\lambda_i)(\tau - \lambda_i) - G(\lambda_i,c_i,\epsilon_i,r)}{r} & \text{if } (b(\lambda_i) - r)(\tau - \lambda_i) \geq G(\lambda_i,c_i,\epsilon_i,r)\\
        -\infty & \text{otherwise }
    \end{cases}
\end{equation}
        
\end{definition}

 We allow $\beta^-_i$ and $\beta^+_i$ to take infinite values to represent scenarios where no amount of external delegation can prevent an SPO from deviating from stake pool operation. The following lemma formalizes how pool deficit and capacity serve as lower and upper bounds to the external delegation an SPO can bear while being content as an SPO. 

\begin{lemma}
\label{lemma:deficit-capacity-stability-pool}
Suppose that the $i$-th player is an SPO with pledge, $\lambda_i$, and pool operation cost, $c_i$, and that they are running a feasible pool under the joint strategy profile $\vect{p}$ with external delegation $\beta_i$. Furthermore, suppose that per-unit delegation rewards in $\vect{p}$ are given by $r$. The $i$-th player prefers operating their pool to becoming idle or becoming a delegator if and only if:
$$
0 < \beta^-_i\leq \beta_i \leq \beta^+_i
$$ 
\end{lemma}

\begin{proof}
The result follows from unpacking $b(\lambda_i) \beta'_i - r\beta_i$ as a piecewise linear expression (due to the piecewise linear nature of $\beta'_i$ resulting from the pool cap $\tau$) in Lemma \ref{lemma:SPO-content} which we recall says that the SPO cannot benefit from deviating from operating their pool if the following holds:
$$
b(\lambda_i)\beta'_i - r\beta_i \geq G(\lambda_i,c_i,\epsilon_i,r) > 0,
$$
where $\beta'_i = \min\{\beta_i,\tau- \lambda_i\}$. For the sake of this proof, we let $h(\beta_i) = b(\lambda_i)\beta'_i - r\beta_i$ and express it piecewise:
\begin{equation}
h(\beta_i)=
    \begin{cases}
        (b(\lambda_i) - r) \beta_i & \text{if } \beta_i \leq \tau - \lambda_i \\
        b(\lambda_i)(\tau - \lambda_i) - r\beta_i & \text{if } \beta_i > \tau - \lambda_i
    \end{cases}
\end{equation}
Considering the gap, $G$, as a value which is independent of $\beta_i$, the condition we seek for an SPO to not deviate is thus:
$$
h(\beta_i) \geq G > 0
$$

We recall that $b(\lambda_i) \geq 0$ for all values of $\lambda_i$ (SPOs never pay the system to open a pool), hence if $(b(\lambda_i) - r) < 0 < G$, then $h(\beta_i)$ is in fact monotonically decreasing in $\beta_i$. Thus, there will be no values of $\beta_i$ such that $h(\beta_i) > G$, which from Lemma \ref{lemma:SPO-content}, implies the SPO will prefer to deviate from operating the pool. Moreover, we notice that $h(\tau - \lambda_i) = (b(\lambda_i) - r)(\tau - \lambda_i) < 0$, hence the expressions for deficit and capacity of the pool give us $\beta^-_i = \infty$ and $\beta^+_i = -\infty$, which also reflects the fact that there exist no value of $\beta_i$ such that $\beta^-_i \leq \beta_i \leq \beta^+_i$. 

When $(b(\lambda_i) - r) > 0$, it follows that the piecewise linear $h(\beta_i)$ is strictly increasing for $\beta_i \in [0,\tau - \lambda_i]$ and strictly decreasing for $\beta_i > \tau - \lambda_i$. As a consequence, the global maximum of $h(\beta_i)$ is at $\beta_i = (\tau - \lambda_i)$. If $h(\tau - \lambda_i) < G$, then $h(\beta_i) \leq h(\tau - \lambda_i) < G$ for all $\beta_i$, hence no amount of external delegation can prevent the SPO from deviating. Moreover, the expression for deficit and capacity are such that once more $\beta^-_i = \infty$ and $\beta^+_i = -\infty$, which also reflect the fact that there exist no value of $\beta_i$ such that $\beta^-_i \leq \beta_i \leq \beta^+_i$.

Finally, if $h(\tau - \lambda_i) > G$, there do exist $\beta_i$ values such that $h(\beta_i) > G$ which prevent the SPO from deviating to delegation or solo pool operation. The expression for $\beta^-_i$ and $\beta^+_i$ have been chosen such that $\beta^-_i \leq \beta^+_i$ and $h(\beta^-_i) = h(\beta^+_i) = G$, where $0 < \beta^-_i$ due to the fact that $G > 0$. Given the piecewise linear nature of $h$, it follows that for $\beta_i \in [\beta^-_i,\beta^+_i]$ we have $h(\beta_i) > G$ as desired.  

\begin{observation}
    Notice that $\beta^-_i \leq \beta_i \leq \beta^+_i$ also implies that the pool opened by the $i$-th player as an SPO is feasible. If this were not the case, then by Lemma \ref{lemma:feasible-pool-structural-lemma} the SPO would prefer delegation, which is not possible due to Lemma \ref{lemma:deficit-capacity-stability-pool}. 
\end{observation}

\end{proof}

\subsubsection{Putting Everything Together}
We summarize the collection of results from this section as a theorem that characterizes useful sufficient conditions for a joint strategy profile, $\vect{p}$, to be a pure Nash equilibrium in a proper delegation game. 

\begin{theorem}\label{thm:PNE-sufficient}
    Suppose that $\cG(\rho,\tau,(\vect{s},\vect{c},\vect{\epsilon}))$ is a proper delegation game. Consider a joint strategy profile $\vect{p}$ that results in per-unit delegation rewards, $r$. The following are sufficient conditions for $\vect{p}$ to be a pure Nash equilibrium:
    \begin{itemize}
        %\item All pools are feasible
        \item Delegators only delegate to feasible pools.
        \item If the $i$-th agent is not idle, they earn at least $\epsilon_i$ utility.
        \item If the $i$-th agent is idle, their delegation utility is at most $\epsilon_i$.
        \item If the $i$-th agent is an SPO with pledge $\lambda_i = s_i < \tau$ and external delegation $\beta_i$, then $\beta^-_i \leq \beta_i \leq \beta^+_i$.
    \end{itemize}
\end{theorem}

\section{The Bayesian Setting}
\label{sec:bayesian-setting}

In a proper delegation game, we let the {\it type} of the $i$-th player consist of their stake, pool operation cost and idle utility: $(s_i,c_i,\epsilon_i)$. In a Bayesian framework we independently draw player types from a common known prior distribution $\cX$ and subsequently have them play a proper delegation game. 
%Thus for a given number of players, $n$, nature first independently draws $(\vect{s},\vect{c},\vect{\epsilon}) \sim \cX^n$ and players subsequently play a given proper delegation game with their types. 
%In more detail, we consider the setting in which we fix a proper reward function, $\rho$, with cap, $\tau$, and players play $\cG(\rho,\tau,(\vect{s},\vect{c},\vect{\epsilon}))$ with $(\vect{s},\vect{c},\vect{\epsilon}) \sim \cX^n$ drawn from a common known prior. 

\begin{definition}[Bayesian Proper Delegation Game]
    A Bayesian proper delegation game requires four inputs: 
    \begin{itemize}
        \item A proper reward function: $\rho$
        \item A pool cap: $\tau$
        \item A type distribution: $\cX$
        \item The number of agents to be drawn from the type distribution: $n > 0$
    \end{itemize}
    For such a game, player types are first drawn independently via $(\vect{s},\vect{c},\vect{\epsilon}) \sim \cX^n$, and they subsequently play the proper delegation game $\cG(\rho,\tau,(\vect{s},\vect{c},\vect{\epsilon}))$. We use the notation $\cG(\rho,\tau, \cX,n)$ to denote a specific Bayesian proper delegation game.
\end{definition}

In Bayesian games one typically studies {\it ex ante} player strategies that consist of mappings from player types to actions taken. Agents in a proper delegation games however have a rich (infinite in fact) family of actions at their disposal. Moreover, as mentioned in the introduction, we are ultimately interested in the high level decision taken by an agent whether to be an SPO, a delegator or idle. For this reason, we introduce the notion of a partial ex ante strategy which will be an important object of study of our paper. 

\begin{definition}[Partial Ex Ante Strategy]\label{def:partial-ex-ante}
    A partial ex ante strategy for a Bayesian delegation game is a function $f: \mathbb{R}^3 \rightarrow \{0,1\}$ which dictates which players become SPOs. Under $f$, a player with type $(s,c,\epsilon)$ is an SPO if and only if $f(s,c,\epsilon) = 1$.
    %$f$ takes as input a player's stake and pool operation cost, $(\lambda,c)$, and represents an indicator for whether a player of said type will be an SPO. More specifically, if a player type is such that $f(s,c) = 1$ they play $a_{SPO}$. On the other hand, if $f(s,c) = 0$, the player will play the better option between being a delegator to a viable pool or remaining idle.   
\end{definition}

The reason we call such ex-ante strategies {\it partial} is due to the fact that after drawing player types, there are multiple pure strategy profiles of the ex post proper delegation game which are consistent with $f$. For a given draw of player types, $(\vect{s},\vect{c},\vect{\epsilon})$, we let $\cA_f(\vect{s},\vect{c},\vect{\epsilon})$ denote the set of pure strategy profiles of the ex post proper delegation game, $\cG(\rho,\tau,(\vect{s},\vect{c},\vect{\epsilon}))$, that are consistent with $f$. In other words, $\vect{p} \in \cA_f(\vect{s},\vect{c},\vect{\epsilon})$ when $p_i = a_{SPO} \iff f(s_i,c_i,\epsilon_i) = 1$. We are ultimately interested in strategies that can give rise to PNE ex post, which are rigorously defined below:

\begin{definition}[Ex post SPO stable]
    Suppose that $f$ is a partial ex ante strategy for a Bayesian proper delegation game $\cG(\rho,\tau, \cX,n)$. We say that $f$ is ex post SPO stable for the draw $(\vect{s},\vect{c},\vect{\epsilon}) \sim \cX^n$ if there exists a joint strategy profile $\vect{p} \in \cA_f(\vect{s},\vect{c},\vect{\epsilon})$ which is a PNE. 
\end{definition}

The main result of this section provides useful sufficient conditions for a partial ex ante strategy, $f$, to be ex post SPO stable for a given draw of player types. Before delving into the main theorem though, we define some relevant quantities. 

\begin{definition}[Total Ex Post Stable Delegation]
    Suppose that $f$ is a partial ex ante strategy for a Bayesian proper delegation game $\cG(\rho,\tau, \cX,n)$ with player types given by $(\vect{s},\vect{c},\vect{\epsilon}) \sim \cX^n$. Assuming that $s^* = \max \{ i \in [n] \mid f(s_i,c_i,\epsilon_i) = 0 \text{ and }\alpha(s_i,c_{min}) \geq \epsilon_i/s_i\}$ and $r = \alpha(s^*,c_{min})$,\footnote{If $\{ i \in [n] \mid \alpha(s_i,c_{min}) \geq \epsilon_i/s_i\} = \emptyset$, we let $r = 0$.}  we denote the total ex post stable delegation by $Del(f)$ and define it by:
    $$
    Del(f) = \sum_{i=1}^n s_i (1 - f(s_i,c_i,\epsilon)) \mathbb{I}(rs_i \geq \epsilon_i) 
    $$
    where $\mathbb{I}(\cdot)$ is an indicator function. 
\end{definition}

\begin{definition}[Total Ex Post Pool Deficit/Capacity]
    Suppose that $f$ is a partial ex ante strategy for a Bayesian proper delegation game $\cG(\rho,\tau, \cX,n)$ with player types given by $(\vect{s},\vect{c},\vect{\epsilon}) \sim \cX^n$. Assuming that $s^* = \max \{ i \in [n] \mid f(s_i,c_i,\epsilon_i) = 0 \text{ and }\alpha(s_i,c_{min}) \geq \epsilon_i/s_i\}$ and $r = \alpha(s^*,c_{min})$, we denote the total ex post pool deficit/capacity by $Def(f)$ and $Cap(f)$ respectively, and define them by:
    $$
    Def(f) = \sum_{i=1}^n\beta^-_i(s_i,c_i,\epsilon_i,r) f(s_i,c_i,\epsilon_i)
    $$
    $$
    Cap(f) = \sum_{i=1}^n\beta^-_i(s_i,c_i,\epsilon_i,r) f(s_i,c_i,\epsilon_i)
    $$
\end{definition}
With the notation above in hand, we can finally prove the main result of this section: 

\begin{theorem}
\label{thm:ex-post-SPO-stable-sufficient}
    Suppose that $f$ is a partial ex ante strategy for a Bayesian proper delegation game $\cG(\rho,\tau, \cX,n)$ with player types given by $(\vect{s},\vect{c},\vect{\epsilon}) \sim \cX^n$. The following is a sufficient condition for $f$ to be ex post SPO stable: 
    $$
    0 < Def(f) \leq Del(f) \leq Cap(f)
    $$
\end{theorem}

\begin{proof}
    Suppose that $f$ satisfies the desired inequalities for a given draw of player types $(\vect{s},\vect{c},\vect{\epsilon}) \sim \cX^n$. We begin with corner cases, the first being when $f(s_i,c_i,\epsilon_i) = 0$ for all players. In this case $Def(f) = Del(f) = Cap(f) = 0$, which satisfies the inequalities of the theorem statement. In addition, such a scenario implies that there are no active pools, hence any form of delegation forcibly earns no utility. This means that the only joint strategy profile $\vect{p} \in \cA_f(\vect{s},\vect{c},\vect{\epsilon})$ which is a PNE is that where all players are idle, hence $f$ is still ex post SPO stable for the draw of player types and the statement holds. Going forward, we assume that there is at least one player with $f(s_i,c_i) = 1$.

    The second corner case occurs when for every player such that $f(s_i,c_i) = 0$ we have $\alpha(s_i,c_{min}) < \epsilon_i/s_i$. Consider any joint strategy profile $\vect{p} \in \cA_f(\vect{s},\vect{c},\vect{\epsilon})$ where the set of delegating agents is non-empty. In this case, there is a pivotal delegate $s^*$ who necessarily earns $\alpha(s^*,c_{min}) s^*$, which from assumption must be less than $\epsilon^*$, their idle utility. It follows that $\vect{p}$ cannot be an ex post PNE. As a consequence, any joint strategy profile $\vect{p} \in \cA_f(\vect{s},\vect{c},\vect{\epsilon})$ which is a PNE must have no delegators, which means that $Del(f) = 0$ and if the $i$-th player is an SPO, it must be the case that $\beta_i = 0$. From Lemma \ref{lemma:deficit-capacity-stability-pool} we know that if the $i$-th agent is an SPO, then their deficit is given by $\beta^-_i > 0$, which cannot be satisfied by $\beta_i = 0$, as a consequence the $i$-th player prefers to deviate from being an SPO and hence $\vect{p}$ is not a PNE. This shows that there can be no $\vect{p} \in \cA_f(\vect{s},\vect{c},\vect{\epsilon})$ which is a PNE for this corner case, and this is consistent with the theorem statement as $Del(f) = 0$ yet $Def(f) > 0$.  

    With the second corner case taken care of, we can make the further assumption that there exists some player such that $f(s_i,c_i,\epsilon_i) = 0$ and $\alpha(s_i,c_{min}) \geq \epsilon_i/s_i$. 
    %We will exhibit a specific joint strategy profile $\vect{p} \in \cA_f(\vect{s},\vect{c},\vect{\epsilon})$ which is a PNE of $\cG(\rho,\tau,(\vect{s},\vect{c},\vect{\epsilon}))$. To do so, we first specify the actions taken by non SPOs, i.e. players such that $f(s_i,c_i) = 0$. 
    %
    Before continuing, let $s^* = \max \{ i \in [n] \mid f(s_i,c_i,\epsilon_i) = 0 \text{ and }\alpha(s_i,c_{min}) \geq \epsilon_i/s_i\}$ and $r = \alpha(s^*,c_{min})$. Moreover, let $A = \{i \in [n] \mid f(s_i,c_i,\epsilon_i) = 0 \text{ and } r s_i < \epsilon_i\}$ and $B = \{i \in [n] \mid f(s_i,c_i,\epsilon_i) = 0 \} \setminus A$. We will show that there exists a PNE, $\vect{p} \in \cA_f(\vect{s},\vect{c},\vect{\epsilon})$, such that if $i \in A$, the $i$-th agent is idle ($p_i = a_{I}$) and if $i \in B$, the $i$-th agent is a delegator ($p_i \in \cD_i$). In such a strategy profile, it must be the case that $s^*$ is the pivotal stake and $r$ is the per-unit delegation rewards to feasible pools.
    
    For now let us assume that all delegation is given to feasible pools (we will show this is possible shortly). If the $i$-th player is a delegator, then $i \in B$, in which case the agent earns $r s_i \geq \epsilon_i$, hence they weakly prefer being a delegator to being idle. 
    
    If the $i$-th player is idle, we distinguish two potential cases. The first case is when $s_i < s^*$, in which case if they agent deviates to becoming a delegator, they stand to earn $r s_i$. However, the fact that the agent is idle implies that $i \in B$, in which case $r s_i < \epsilon_i$. The second case is when $s_i > s^*$, in which case the construction of $s^*$ implies that $\alpha(s_i,c_{min}) < \epsilon_i/s_i$. If such a player deviates to becoming a delegator, doing so changes per-unit delegation rewards to $\alpha(s_i,c_{min})$ in which case they earn $\alpha(s_i,c_{min}) s_i < \epsilon_i$ utility for doing so, which is less than what they obtain from being idle. 

    To finalize the proof, we notice that if $\vect{p} \in \cA_f(\vect{s},\vect{c},\vect{\epsilon})$ is such that for $i \in A$, the $i$-th agent is idle ($p_i = a_{I}$) and for $i \in B$, the $i$-th agent is a delegator ($p_i \in \cD_i$), it must be the case that the total stake to be delegated is precisely $Del(f)$. In addition, $Def(f)$ and $Cap(f)$ also represent the sum of all pool deficits and capacities, respectively, hence the fact that $Def(f) \leq Del(f) \leq Cap(f)$ implies that there exists a way to delegate to pools that respects individual pool deficits and capacities. The resulting $\vect{p} \in \cA_f(\vect{s},\vect{c},\vect{\epsilon})$ from delegating this way is in turn a PNE from Theorem \ref{thm:PNE-sufficient} as desired.

\end{proof}

If $f$ is ex post SPO stable for the draw $(\vect{s},\vect{c},\vect{\epsilon}) \sim \cX^n$ there are generally multiple joint strategy profiles $\vect{p} \in \cA_f(\vect{s},\vect{c},\vect{\epsilon})$ which give rise to PNE. In the following section we provide a means of distinguishing the performance different PNE which arise. We quantify performance of a given joint strategy profile $\vect{p}$ via 3 key metrics: Decentralization, Participation and System Expenditure.
%The precise definition of these metrics for any given joint strategy profile, $\vect{p}$ is given in the following section. 

\section{ Decentralisation, Participation and Expenditure Objectives}
\label{sec:objectives}
\subsection{Decentralization Objective}

 Recall that a specific strategy profile, $\vect{p} \in \cA$, consists of relevant information regarding which agents have activated pools, which agents have delegated to said active pools, and which agents forego participating in the pool creation/delegation scheme. From the strategy profile, we can extrapolate the {\em public pool profile}, which consists of the information available to a third-party observer of the system (who may not know which agent specifically owns stake used to pledge or delegate). We encode the public profile with two vectors,  $(\vect{\lambda},\vect{\beta})$, of variable dimension $1 \leq k \leq  n$ which in turn represents the number of pools that are active in a public profile. For a given pool $j \in [k]$, the terms $\lambda_j$ and $\beta_j$ represent how much was pledged to open the pool and how much external stake is delegated to the pool respectively.  In addition, $\sigma_j = \lambda_j + \beta_j$ is the size of the $j$-th pool, so that $\vect{\sigma} = \vect{\lambda} + \vect{\beta}$ is a vector containing the sizes of all pools created in a strategy profile. With this notation on hand we can define the following objectives that measure the relative performance of different joinst strategy profiles in a proper delegation game:

 \subsection{Participation Objective}
 In order to evaluate the participation of a system we compute the sum of  the \textbf{absolute stake} that is either delegated or pledged (a quantity which we call the ``active stake''). A system designer seeks to maximize participation.
 \begin{definition}[Participation Objective]
 \label{def:participation-objective}
 Let $\vect{p} \in \cA$ be a joint strategy profile in the proper delegation game, $\cG(\rho,\tau,(\vect{s},\vect{c},\vect{\epsilon}))$, that gives rise to the public pool profile $(\vect{\lambda},\vect{\beta})$ with $k$ pools of sizes given by $\vect{\sigma} = \vect{\lambda} + \vect{\beta}$.
 We define the participation objective $O^P$ as follows:
 $$
 O^P(\vect{p}) = \sum_{j=1}^k (\lambda_j+\beta_j)= \sum_{j=1}^k \sigma_j
 $$
% When the context is clear we write $T$ instead of $T(\vect{p})$ for simplicity. 
\end{definition}
 
\subsection{Expenditure Objective}
We evaluate the cost that is incurred by the system in paying all agents for their participation in the system as design objective. Unlike participation, a system designer ideally seeks to minimize expenditure. 
\begin{definition}[Expenditure Objective]
\label{def:expenditure-objective}
Suppose that $\vect{p} \in \cA$ is a joint strategy profile for the proper delegation game, $\cG(\rho,\tau,(\vect{s},\vect{c},\vect{\epsilon}))$. We define the expenditure objective, $O^E$ as follows:
$$
O^E(\vect{p}) = \sum_{i=1}^n R_i(\vect{p})
$$
\end{definition}

\subsection{Decentralization Objective}
Finally we define a family of decentralization objectives $O^D_\ell$, with relevant parameter $\ell \geq 0$. For a fixed parameter, $\ell$, $O^D_\ell$ takes as input a joint strategy profile $\vect{p} \in \cA$ in the proper delegation game, $\cG(\rho,\tau,(\vect{s},\vect{c},\vect{\epsilon}))$ and outputs the smallest collective pledge amongst coalitions of pools of aggregate size exceeding an $\ell \cdot O^P(\vect{p})$. The value of $\ell$ will typically take values of relevance to resilience guarantees in Byzantine consensus protocols (i.e. $1/3,1/2,2/3$). The following is a more precise definition.

\begin{definition}[Decentralization Objective]
\label{def:decentralization-objective}
Suppose that $\vect{p} \in \cA$ is a joint strategy profile in the proper delegation game, $\cG(\rho,\tau,(\vect{s},\vect{c},\vect{\epsilon}))$, with a public pool profile given by $(\vect{\lambda},\vect{\beta})$ over $k$ pools. For a given $\ell \geq 0$, we let $P_\ell(\vect{p})$ denote the set of pool coalitions with aggregate stake exceeding $\ell \cdot O^P(\vect{p})$:

$$
P_\ell(\vect{p})=\{S \subseteq [k]: \sum_{i\in S} \sigma_i \geq \ell \cdot  O^P(\vect{p})\}.
$$
  With this in hand, we define the decentralization objective $O^D_\ell(\vect{p})$ as follows:
$$
O^D_{\ell}(\vect{p}) = \min_{S\in P_\ell(\vect{\lambda},\vect{\beta})} \sum_{i\in S} \lambda_i.
$$
\end{definition}
Notice that most of our definitions do not preclude us from considering a scenario in which all agents forego participating in the protocol. In this case, $k = 0$, and $\vect{\lambda},\vect{\beta} = \{0\}$, the unique zero-dimensional vector. Furthermore $P_\ell(\vect{\lambda},\vect{\beta}) = \emptyset$ as $[0] = \emptyset$, and the decentralization objective of this strategy profile is 0.

\subsection{Multi-objective Optimization}
In all that follows of this paper, we will be interested in measuring the performance of payment schemes for delegation games over the the three objectives mentioned above. As mentioned previously, a system designer will seek to maximize participation, minimize expenditure and maximize decentralization. Simultaneously optimizing for each of these objectives is generally not possible, and hence we use a framework inspired by multi-objective optimization to understand tradeoffs between all three.

\section{Computational Methods and Results}
\label{sec:computational-methods-results}

Our main computational approach focuses on conceptualizing the performance of a partial ex ante strategy, $f$, for a given Bayesian proper delegation game $\cG(\rho,\tau,\cX,n)$. To do so, we measure the performance of $f$ for a given draw of player types, $(\vect{s},\vect{c},\vect{\epsilon}) \sim \cX^n$, in terms of the three objectives from Section \ref{sec:objectives}. At a high level, our approach proceeds in two stages: 
\begin{enumerate}
    \item First we establish whether $f$ satisfies the sufficient conditions set forth in Theorem \ref{thm:ex-post-SPO-stable-sufficient} for being ex post SPO stable.
    \item If $f$ is ex post SPO stable, then all $\vect{p} \in \cA_f(\vect{s},\vect{c},\vect{\epsilon})$ which are PNE exhibit the same participation breakdown (the amount of stake which is dedicated to being idle, delegating or pledging as an SPO respectively), and hence have equal values for $O^P$. This is not the case for $O^E$ and $O^D_\ell$, hence to study decentralization and expenditure, we construct a comprehensive set of ex post PNE, $\vect{p}^1,\dots,\vect{p}^m \in \vect{P} \in \cA_f(\vect{s},\vect{c},\vect{\epsilon})$ with different decentralization and expenditure performance to represent the potential spread of performance that can be achieved ex post for $f$. 
\end{enumerate}

\subsection{Representative Ex Post PNE}

In what follows we outline our methodology for constructing a representative set of PNE from $\cA(\vect{s},\vect{c},\vect{\epsilon})$ for understanding the potential decentralization and expenditure achieved by a given partial ex ante strategy, $f$, which is ex post SPO stable for a given draw of agent types. 

We consider a Bayesian proper delegation game, $\cG(\rho,\tau,\cX,n)$ and a partial ex ante strategy, $f$. Suppose that $f$ is ex post SPO stable for a given draw of player types, $(\vect{s},\vect{c},\vect{\epsilon})$, where at least one agent is an SPO. In what follows we outline our methodology for constructing a representative set of PNE from $\cA(\vect{s},\vect{c},\vect{\epsilon})$ for understanding the potential decentralization and expenditure achieved under $f$ ex post.

We let $\lambda_{min} \leq \lambda_{max}$ represent the smallest and largest pledges made by SPOs under $f$. More specifically, 
$$
\lambda_{min} = \min_{i: f(s_i,c_i,\epsilon_i) = 1} s_i \leq \max_{i: f(s_i,c_i,\epsilon_i) = 1} s_i = \lambda_{max}.
$$
We also let $m \in \mathbb{N}$ be a resolution parameter that dictates the number of representative PNE from $\cA_f(\vect{s},\vect{c},\vect{\epsilon})$ constructed. From these quantities, we construct an $m$-dimensional vector of {\it reference pledges}, $\bar{\vect{\lambda}} = (\bar{\lambda}_j)_{j=1}^m$, where the $j$-th reference pledge is defined as follows:
$$
\bar{\lambda}_j = \lambda_{min} + (j-1)\frac{(\lambda_{max} - \lambda_{min})}{m-1}
$$  

With $\bar{\lambda}_j$ in hand, we can construct the $j$-th representative PNE from $\cA_f(\vect{s},\vect{c},\vect{\epsilon})$ which we denote by $\vect{p}^j$. As in Theorem \ref{thm:ex-post-SPO-stable-sufficient}, we can fix the high level actions of agents between remaining idle to ensure ex post SPO stability. To do so, we once more let $s^* = \max \{i \in [n] \mid f(s_i,c_i,\epsilon_i) = 0 \text{ and } \alpha(s_i,c_i) \geq \epsilon_i/s_i\}$ and we let $r = \alpha(s^*,c_{min})$. We now consider an arbitrary $i$-th player in $\cG(\rho,\tau,(\vect{s},\vect{c},\vect{\epsilon}))$:
\begin{itemize}
\item If $f(s_i,c_i,\epsilon_i) = 0$ and $r s_i < \epsilon_i$, then $\vect{p}^j_i = a_{I}$
\item If $f(s_i,c_i,\epsilon_i) = 0$ and $r s_i \geq \epsilon_i$, then $\vect{p}^j_i \in \cD_i$
\item If $f(s_i,c_i,\epsilon_i) = 1$, then $\vect{p}^j_i = a_{SPO}$
\end{itemize}

All that remains to specify $\vect{p}^j$ is deciding where delegation goes to, for which we make use of the reference pledge, $\bar{\lambda}_j$. We do so by computing a delegation vector $\vect{\beta} = (\beta_i)_{i=1}^n$ first satisfying the deficit of all pools (using $Def(f) \leq Del(f)$ of the available delegation). Afterwards, we greedily fill pools with pledge closest to $\bar{\lambda}_j$ up to capacity using the remaining $Del(f) - Def(f)$ delegation at our disposal. The details of the greedy delegation allocation are provided in Algorithm \ref{alg:greedy-delegation}. Given the target greedy delegation allocation, $\vect{\beta}$, we simply let $\vect{p}^j$ be any PNE which is consistent with the target delegation (since they all achieve the same expenditure and decentralization objectives). 

\begin{algorithm}
    \caption{Greedy Delegation Allocation}
    \label{alg:greedy-delegation}
    \begin{algorithmic}[1] % The number tells where the line numbering should start
        \Procedure{GreedyDelegation}{$\bar{\lambda}_j,\vect{\beta^-},\vect{\beta^+},Del(f)$} 
        %\Comment{The g.c.d. of a and b}
            \State $\vect{\beta}\gets \vect{\beta^-}$
            \Comment{Satisfying pool deficit}
            \State $X \gets Del(f) - \sum_{i=1}^n \beta_i$
            \Comment{Remaining delegation}
            \State $A \gets \{i \in [n] \mid \beta_i < \beta^+_i\}$
            \State $j^* \gets \text{argmin}_{i \in A} |\lambda_i - \bar{\lambda}_j|$
            \Comment{Ties broken lexicographically in argmin}
            \While{$X\not=0$} 
            %\Comment{While there remains delegation to be allocated}
                \State $ \beta_{j^*}\gets \beta_{j^*} + \min\{X,(\beta^+_{j^*} - \beta_{j^*}) \}$
                \State $X \gets Del(f) - \sum_{i=1}^n \beta_i$
                \State $A \gets \{i \in [n] \mid \beta_i < \beta^+_i\}$
                \State $j^* \gets \text{argmin}_{i \in A} |\lambda_i - \bar{\lambda}_j|$
                %\Comment{Ties broken lexicographically in argmin}
            \EndWhile\label{euclidendwhile}
            \State \textbf{return} $\vect{\beta}$
            %\Comment{The gcd is b}
        \EndProcedure
    \end{algorithmic}
\end{algorithm}

\paragraph*{Computing Participation and Expenditure Objectives}

Computing $O^P$ and $O^E$ for a given $\vect{p} \in \cA$ in a proper delegation game, $\cG(\rho,\tau,(\vect{s},\vect{c},\vect{\epsilon}))$, is straightforward. In order to do so, we extrapolate the relevant public pool profile, $(\vect{\lambda},\vect{\beta})$ for $\vect{p}$, where $\vect{\lambda} = (\lambda_j)_{j=1}^k$ and $\vect{\beta} = (\beta_j)_{j=1}^k$ represent the pledge and external delegation that arise for the $k \geq 0$ active pools. As per Definitions \ref{def:participation-objective} and \ref{def:expenditure-objective}, the participation and expenditure objectives are given by:
$$
O^P(\vect{p}) = \sum_{j=1}^k (\lambda_j+\beta_j)
$$
$$
O^E(\vect{p}) = \sum_{i=1}^n R_i(\vect{p})
$$
In the scenario where all pools from $\vect{p}$ are feasible, it is the case that the utility an SPO earns is given by $\rho(\lambda_j,\beta_j) - r \beta_j - c > 0$. Moreover, the total rewards given to delgators to the pool is $r \beta_j$, hence when summing rewards given to all agents in the system, it suffices to compute the sum of rewards over pools, hence we get
$$
O^E(\vect{p}) = \sum_{j=1}^k \rho(\lambda_j,\beta_j)
$$ 

\paragraph*{Approximating the Decentralization Objective}

To wrap up our computational methods, we focus on the problem of computing the decentralization objective, $O^D_\ell$, for a given joint strategy $\vect{p} \in \cA$ in a given proper delegation game, $\cG(\rho,\tau,(\vect{s},\vect{c},\vect{\epsilon}))$. As per Definition \ref{def:decentralization-objective}, the value of $O^D_\ell(\vect{p})$ is the smallest cumulative stake of any coalition of pools with size that exceeds $\ell T$. We can express this computational problem in terms of the public pool profile $(\vect{\lambda},\vect{\beta})$ which arises from $\vect{p}$. To do so, we let $\vect{\lambda} = (\lambda_j)_{j=1}^k$, $\vect{\beta} = (\beta_j)_{j=1}^k$ and $\vect{\sigma} = \vect{\lambda} + \vect{\beta}$ represent the pledge, external delegation and total size of each of the $k \geq 0$ active pools that arise from $\vect{p}$. With this in hand, the value of $O^D_{\ell}(\vect{p})$ is given by the optimization problem in Equation \ref{eq:decentralization-optimization}. 

\begin{equation}
\label{eq:decentralization-optimization}
\begin{aligned}
\min_{x_1,\dots x_k} \quad & \sum_{j=1}^{k}{\lambda_j x_j}\\
\textrm{s.t.} \quad & \sum_{j=1}^k \sigma_j x_j \geq \ell T\\
  &x_j \in \{0,1\}    \\
\end{aligned}
\end{equation}

This optimization problem is NP-hard as it is precisely an instance of the $\{0,1\}$-min knapsack problem, \cite{csirik1991heuristics}. In order to approximate $O^D_\ell$, we use the typical dynamic programming FPTAS as per \cite{tauhidul2009approximation}. 

\subsection{Relevant Modeling Choices and Parameters}

In this section we provide details regarding further modeling choices and parameter settings we make before delving into experimental results. 

\paragraph*{Threshold Partial Ex Ante Strategies}

Our framework for partial ex ante strategies is very general. For a given Bayesian proper delegation game, $\cG(\rho,\tau,\cX,n)$, a partial ex ante strategy can be an arbitrary function from player types to whether they act as an SPO or not. In practice we expect larger players (with more stake) to be SPOs for multiple reasons (increased interest in the proper functioning of the underlying blockchain, potentially less frictions to operate as SPO, etc.). For this reason, we consider a simple class of partial ex ante strategies with agents operating as SPOs only if they exceed a stake threshold. 
\begin{definition}[Threshold Partial Ex Ante Strategy]
We let $f^t_\alpha: \mathbb{R}^2 \rightarrow \{0,1\}$ denote a threshold partial ex ante strategy with threshold $\theta \geq 0$. The strategy is specified by:
$$
f^t_\theta(s,c,\epsilon) = 1 \iff s \geq \theta
$$
\end{definition}

\paragraph*{Bounded Pareto Distribution for Stake}

As is common in economic literature, we can assume that stake distributions obey a power law \cite{pareto1964cours}. For this reason, we consider type distributions such that the marginal distribution of stake obeys a bounded Pareto distribution:

\begin{definition}[Truncated Pareto Distribution]\label{def:Pareto}
We say that $Z$ is a Pareto distribution with minimum value $L > 0$, maximum value $H > L$ and inequality parameter $\gamma$ if it has a pdf given by:
\begin{equation*}
\eta(x) = 
 \begin{cases} 
      \left(\frac{\gamma L^\gamma}{1 - (L/H)^\gamma} \right) x^{-\gamma-1} & x \in [L,H] \\
      0 & x \notin [L,H] 
   \end{cases}
\end{equation*}
We write $s \sim Pareto(L,H,\gamma)$ when an agent's stake is distributed according to a bounded Pareto distribution. 
\end{definition}

In order to acheive marginal Pareto distributions on player stake, we consider type distributions $\cX$ which result as product distributions over player stake, cost and idle utility respectively. Furthermore, without loss of generality, we normalize the value of stake with respect to the lower bound $L$, so we can let $L = 1$. In more detail, we consider type distributions parametrized by: 
\begin{itemize}
\item $H,\gamma$: the upper bound and exponent in Pareto PDF for stake distribution.
\item $c_{min},c_{max}$: the minimal and maximal values of pool operation cost.
\item $\epsilon_{min},\epsilon_{max}$: the minimal and maximal values of idle utility.
\end{itemize}

The type distribution with these parameters is denoted $\cX(H,\gamma,c_{min},c_{max},\epsilon_{min},\epsilon_{max})$, though when evident from context, we simply use $\cX$ as before. In order to sample from the distribution, $(s,c,\epsilon) \sim \cX(H,\gamma,c_{min},c_{max},\epsilon_{min},\epsilon_{max})$, we independently sample each component $s \sim Pareto(1,H,\gamma)$, $c \sim U[c_{min},c_{max}]$ and $\epsilon \sim U[\epsilon_{min},\epsilon_{max}]$. 

\subsection{Experimental Results} 

We provide some results for a proper Bayesian delegation game which demonstrate the flexibility of our approach in studying tradeoffs struck by payment schemes in proper delegation games. In what follows, we assume a baseline parameter setting upon which we modulate key parameters to show their impact on participation, decentralization, and expenditure objectives. 

\paragraph*{Baseline Parameter Settings}

We begin by providing details regarding the family of $\rho$ functions we explore in our experiments. Given we are modeling proper delegation games as per Definition \ref{def:proper-delegation-game}, we are considering separable pool reward functions such that $\rho(\lambda,\beta) = a(\lambda) + b(\lambda)\beta'$, where $\beta' = \min \{\tau - \lambda, \beta\}$ for the cap $\tau$, which we will specify shortly. In our experiments, we model $a(\lambda)$ and $b(\lambda)$ as polynomials of varying degree and positive coefficients (which is in fact similar to the formula for Cardano reward sharing schemes \cite{brunjes2020reward}). Our baseline formulas are given by $a(\lambda) = b(\lambda) = \lambda$. 

As an aside, we note that if $a(\lambda) = \sum_{i=1}^m z_i \lambda^i$, where $z_i > 0$ for all $i$, then it follows that $\alpha(\lambda,c) = \frac{a(\lambda) - c}{\lambda} = (\sum_{i=1}^m z_i \lambda^{i-1}) - \frac{c}{\lambda}$, which is in fact monotonically increasing in $\lambda$, as is required for a proper delegation game.  

For the marginal distribution of player stakes, we use a truncated Pareto distribution with lower bound $L = 1$, upper bound $H = 100$, and inequality parameter $\gamma = 1.5$. For SPO costs, we let lower and upper bounds for cost be $c_{min} = 0.4$ and $c_{max} = 0.6$ and for idle utilities, we simply assume that all players have the same $\epsilon = 0.01$. Finally, given the marginal stake distribution, we let $\tau = 200$ be the pool cap used for $\rho$. We begin by considering the threshold partial ex ante strategy $f^t_\theta$ with $\theta = 30$. Moreoever, we consider a Bayesian proper delegation game with $n = 1000$ agents drawn from the type distribution described above. In addition, we create $m = 100$ representative ex post PNE as per Algorithm \ref{alg:greedy-delegation} whenever $f^t_\theta$ is ex post SPO stable, and use $\ell = 0.5$ for the decentralization objective $O^d_\ell$. Finally, we repeat this process for $N = 500$ independent draws from $\cX^n$. 

Results from this parameteric setting are presented in Figures \ref{fig:baseline-participation} and \ref{fig:baseline-decentralization-expenditure}. With regards to participation, the empirical frequency of ex post stability for $f^t_\theta$ was 496 of the $N = 500$ draws of player types. In Figure \ref{fig:baseline-participation} we provide a breakdown of the participation achieved by $f^t_\theta$ for these draws, and we note that no players are idle in this setting. The proportional amount of stake used as SPO pledge and delegation respectively varies by about 0.15. With regards to expenditure and decentralization, we turn to Figure \ref{fig:baseline-decentralization-expenditure}, where we can see that in general as delegation is sent to pools with higher pledge, the system achieves better decentralization, albeit at a higher expenditure.

\begin{figure}
    \centering
    \includegraphics[width=0.45\textwidth]{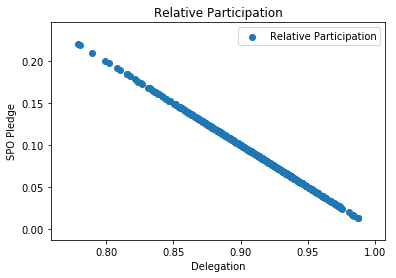}
    \includegraphics[width=0.48\textwidth]{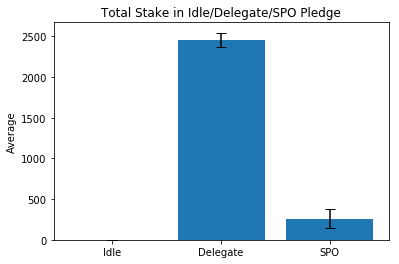}
    \caption{This Figure provides a breakdown of participation for the baseline parameter setting. Each point in the left plot is one of the 496 draws of types in the Bayesian PNE that gave rise to ex post SPO stability. The axes represent the relative proportion of stake that is used for delegation and SPO pledges. As we can see, all points lie on a line indicative of the fact that for no draw do we see idle agents. The right bar chart provides average values of absolute stake used by agents being idle, delegators or SPOs respectively. }
    \label{fig:baseline-participation}
\end{figure}

\begin{figure}
    \centering
    \includegraphics[width=0.45\textwidth]{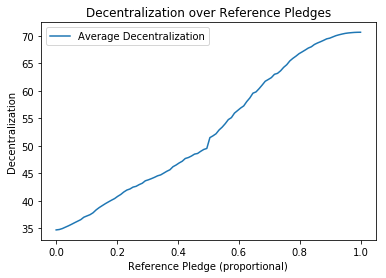}
    \includegraphics[width=0.475\textwidth]{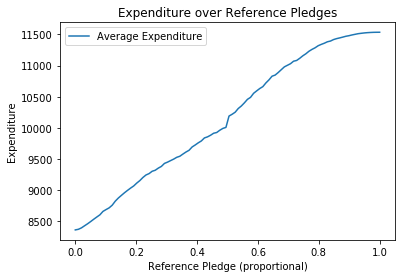}
    \includegraphics[width=0.475\textwidth]{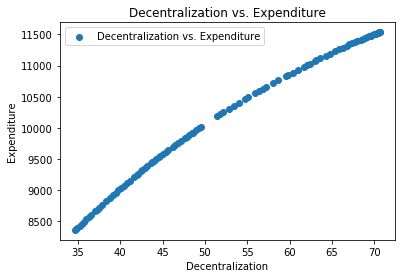}
    \caption{The top two plots provide insight regarding the spread of values the decentralization and expenditure objectives can take for ex post PNE in the baseline parameter setting. The $x$-axis for both of these plots corresponds to different representative PNE as per Algorithm \ref{alg:greedy-delegation}, in which the defining characteristic of a representative PNE is the reference pledge $\bar{\lambda}_j$, which is a proportional value relative to the spread of SPO pledges. The bottom graph simultaneously plots the performance of each representative ex post PNE in terms of decentralization and expenditure.}
    \label{fig:baseline-decentralization-expenditure}
\end{figure}

\paragraph*{Impact of Idle Utility}

In this section we modulate the idle utility: $\epsilon \in \{0.005,0.1,1.0,5.0,10.0\}$ of all players in the game. In Table \ref{table:modulate-eps} we see the empirical frequency of ex post stable PNE as we modulate $\epsilon$ values, and we see that there is no significant difference even as $\epsilon$ increases multiple orders of magnitude. We do however see significant differences in terms of the participation, decentralization and expenditure of ex post PNE as we change idle utilities. With regards to participation, Figure \ref{fig:vary-eps-participation} shows the changes in relative and absolute participation of agents as $\epsilon$ varies. As expected, with higher idle utilities, more agents prefer remaining idle over delegating. Moreover, this is in line with the fact that empirical frequencies for ex post stability do not change much, for if there is less delegation to go around, it can be easier to satisfy pool deficits and capacities. Of course, if too much delegation is idle, then there may not be enough delegation to satisfy pool deficits, and we may see a decrease in the empirical frequency of ex post SPO stability. Finally Figure \ref{fig:vary-eps-decentralization-expenditure} provides insight in terms of how decentralization and expenditure vary with $\epsilon$. As expected, large values of $\epsilon$ result in lower expenditure, as the system needs to pay out less delegators. On the other hand, we also see that larger baseline utilities can increase decentralization, which also makes sense from the decreased delegation that occurs, as any dominating coalition of pools will necessarily have more skin in the game as they may have less external delegation.

\begin{table}
\centering
\begin{tabular}{ |c|c|c|c|c|c| } 
 \hline
 $\epsilon$ & 0.005 & 0.1 & 1.0 & 5.0 & 10.0 \\ 
 \hline
 Ex post SPO stable draws & 498 & 497 & 499 & 495 & 499 \\ 
 \hline
\end{tabular}
\caption{The number of ex post SPO stable draws (out of 500) for different $\epsilon$ values.}
\label{table:modulate-eps}
\end{table}

\begin{figure} 
    \centering
    \includegraphics[width=0.45\textwidth]{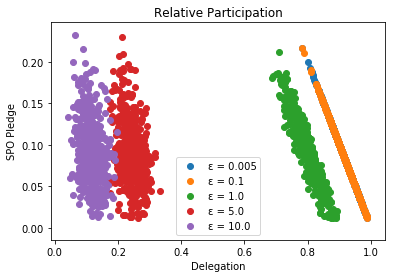}
    \includegraphics[width=0.48\textwidth]{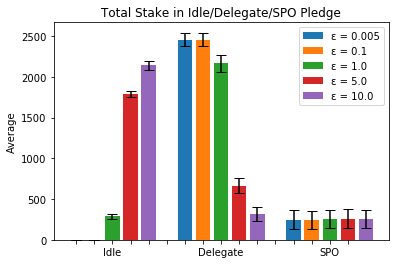}
    \caption{This Figure provides a breakdown of participation as $\epsilon$ varies in $\{0.005,0.01,0.02,0.05\}$. Different $\epsilon$ values to different colors and each point in the plot corresponds to draws of types that gave rise to ex post SPO stability. The axes represent the relative proportion of stake that is used for delegation and SPO pledges. The right bar chart provides average values of absolute stake used by agents being idle, delegators or SPOs respectively for different threshold values.}
    \label{fig:vary-eps-participation}
\end{figure}

\begin{figure}
    \centering
    \includegraphics[width=0.45\textwidth]{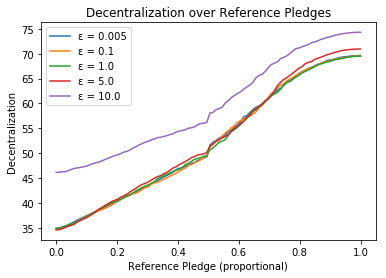}
    \includegraphics[width=0.475\textwidth]{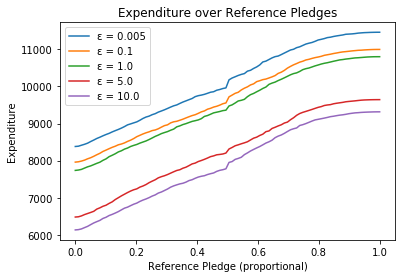}
    \includegraphics[width=0.475\textwidth]{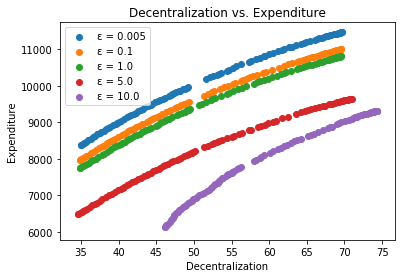}
    \caption{The top two plots provide insight regarding the spread of values the decentralization and expenditure objectives can take for ex post PNE $\epsilon$ values vary in $\{0.005,0.01,0.02,0.05\}$. The $x$ axis for both of these plots correspond to different representative PNE as per Algorithm \ref{alg:greedy-delegation}, in which the defining characteristic of a representative PNE is the reference pledge $\bar{\lambda}_j$, which is a proportional value relative to the spread of SPO pledges. The bottom graph simultaneously plots the performance of each representative ex post PNE in terms of decentralization and expenditure.}
    \label{fig:vary-eps-decentralization-expenditure}
\end{figure}

\paragraph*{Impact of Reward Function}

In this section we modulate the separable reward function we use in the proper delegation game, $\rho(\lambda,\beta) = a(\lambda) + b(\lambda)\beta'$. In addition we fix idle utilities to be larger than baseline at $\epsilon = 5$, where we've seen that agents can prefer to be idle over delegating. In this way we can glean insight regarding how different payment structures can foster participation. We modulate our payment scheme by varying, $a,b$ and $\tau$. Going forward we consider setting the constituent functions of $\rho$ with combinations of the following functions:
\begin{itemize}
    \item $g_1(\lambda) = 0.5\lambda$
    \item $g_2(\lambda) = \lambda$
    \item $g_3(\lambda) = 2\lambda$
    \item $g_4(\lambda) = \lambda + 0.005 \lambda^2$
    \item $g_5(\lambda) = \lambda + 0.01 \lambda^2$
    \item $g_6(\lambda) = \lambda + 0.05 \lambda^2$
\end{itemize}

We modulate $\rho$ in three different ways. First, we unilaterally modulate $a \in \{g_1,\dots,g_6\}$, then we unilaterally modulate $b \in \{g_1,\dots,g_6\}$, and finally we jointly modulate $(a,b) \in \{(g_1,g_1)\dots,(g_6,g_6)\}$. Empirical frequencies of ex post SPO stability are in Table \ref{table:modulate-a-b-ab}.

\begin{table}
\centering
\begin{tabular}{ |c|c|c|c|c|c|c| } 
 \hline
  & $g_1$ & $g_2$ & $g_3$ & $g_4$ & $g_5$ & $g_6$ \\ 
 \hline
 Modulate $a$ & 497 & 498 & 495 & 497 & 489 & 449 \\ 
 \hline
 Modulate $b$ & 497 & 498 & 496 & 495 & 496 & 499 \\ 
 \hline
 Modulate $(a,b)$ & 496 & 496 & 496 & 497 & 499 & 493 \\ 
 \hline
\end{tabular}
\caption{The number of ex post SPO stable draws (out of 500) for different settings of $\rho$.}
\label{table:modulate-a-b-ab}
\end{table}

In Figure \ref{fig:vary-rho-participation} we provide a detailed breakdown of how modulating $a$ and $b$ within $\rho$ can impact the participation reached by the system at ex post PNE. First of all we see that unilaterally modulating $a \in \{g_1,\dots,g_6\}$ (first row of Figure \ref{fig:vary-rho-participation}) accounts for much more change in participation over unilaterally modulating $b \in \{g_1,\dots,g_6\}$ (second row of Figure \ref{fig:vary-rho-participation}). Moreoever, when jointly modulating $(a,b) \in \{(g_1,g_1),\dots,(g_6,g_6)\}$ (third row of Figure \ref{fig:vary-rho-participation}), changes in participation closely resemble those made by individually modulating $a$, which suggest that for the functional values chosen, changes in $a$ account for the majority of differences in participation. This phenomenon largely results from the fact that the $a$ functions we explore with larger quadratic coefficients in $\lambda$ not only pay SPOs more, but they also increase values of $\alpha(s,c)$, which in turn increase delegation rewards. Increased delegation rewards in turn incentivize more players into being delegators over being idle. At the same time, this comes at an added expense, as can be seen in Figure \ref{fig:vary-rho-decentralization-expenditure} where higher degree expressions of $\lambda$ result in higher expenditure for the system. At the same time, these expensive ex post PNE also acheive large decentralization values, hence the system designer may find it beneficial to use such $\rho$ functions if prioritizing participation and decentralization is more important than minimizing expenditure.

\begin{figure} 
    \centering
    \includegraphics[width=0.45\textwidth]{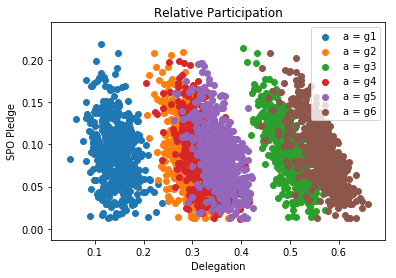}
    \includegraphics[width=0.48\textwidth]{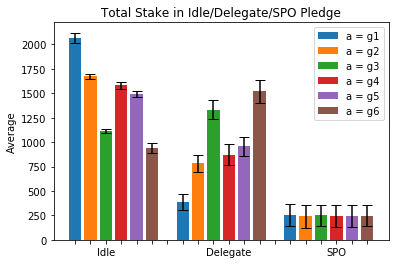}
    \includegraphics[width=0.45\textwidth]{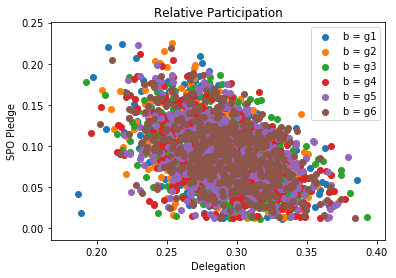}
    \includegraphics[width=0.45\textwidth]{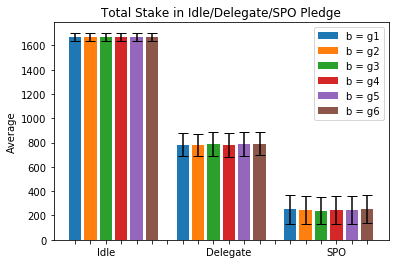}
    \includegraphics[width=0.45\textwidth]{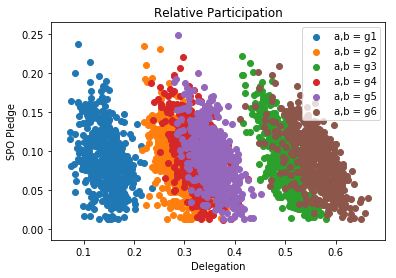}
    \includegraphics[width=0.45\textwidth]{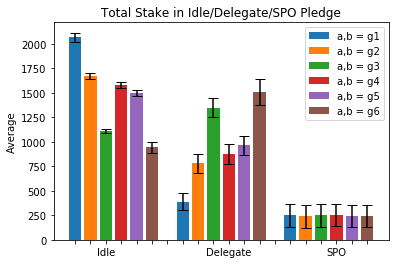}    
    \caption{This Figure provides a breakdown of participation as $a$ and $b$ vary in $\{g_1,\dots,g_6\}$. The first row corresponds to unilaterally modulating $a$, the second row corresponds to unilaterally modulating $b$, and the third row corresponds to modulating $(a,b) \in \{(g_1,g_1),\dots,(g_6,g_6)\}$. For each row, the left image is scatter plot where each point of a given color is an ex post PNE for a given $\rho$ function. For each row, the right image corresponds to the spread of absolute participation of each type (idle, delegation, SPO) for a given $\rho$ function.}
\label{fig:vary-rho-participation}
\end{figure}

\begin{figure}
    \centering
    \includegraphics[width=0.3\textwidth]{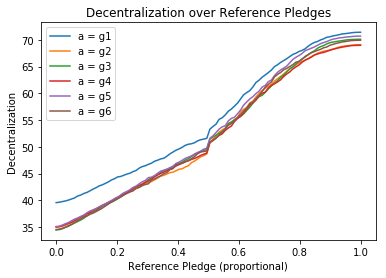}
    \includegraphics[width=0.32\textwidth]{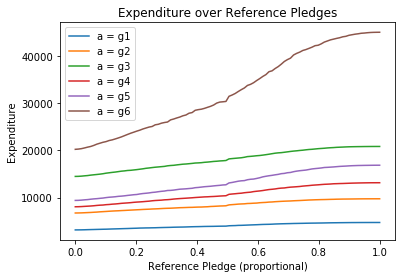}
    \includegraphics[width=0.32\textwidth]{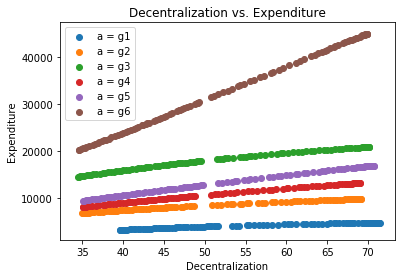}
    \includegraphics[width=0.3\textwidth]{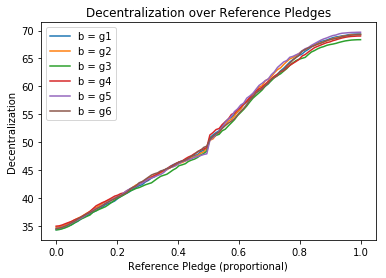}
    \includegraphics[width=0.32\textwidth]{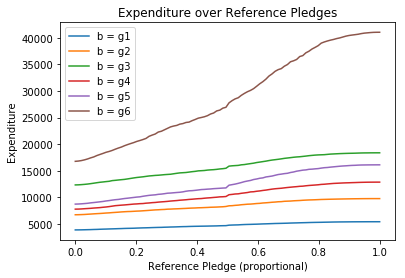}
    \includegraphics[width=0.32\textwidth]{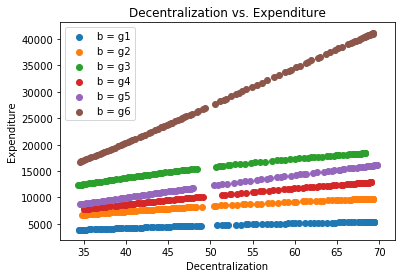}
    \includegraphics[width=0.3\textwidth]{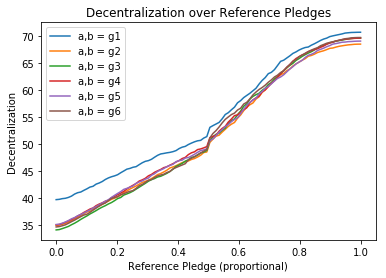}
    \includegraphics[width=0.32\textwidth]{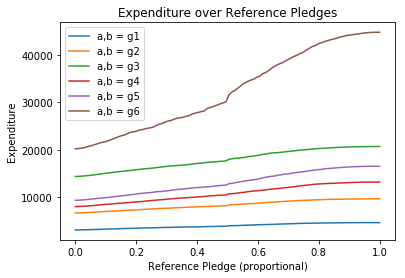}
    \includegraphics[width=0.32\textwidth]{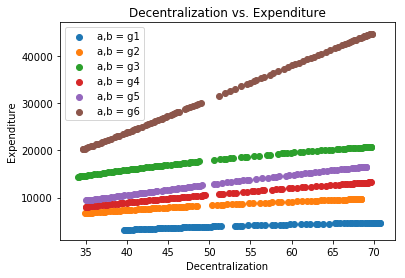}
    \caption{This Figure provides a breakdown of decentralization and expenditure for representative ex post PNE as $a$ and $b$ vary in $\{g_1,\dots,g_6\}$. The first row corresponds to unilaterally modulating $a$, the second row corresponds to unilaterally modulating $b$, and the third row corresponds to modulating $(a,b) \in \{(g_1,g_1),\dots,(g_6,g_6)\}$. For each row, the left image plots decentralization and the middle image expenditure for representative ex post PNE with increasing reference pledge values. For a given row, the right image simultaneously plots decentralization and expenditure for each representative ex post PNE. For each plot, different colors correspond to different $\rho$ functions generated by modulating $a$ and $b$.}
\label{fig:vary-rho-decentralization-expenditure}
\end{figure}

Finally,  we also modulate $\tau \in \{100,150,200,250\}$. Empirical frequencies of ex post SPO stability can be found in Table \ref{table:modulate-tau}. Once more we use $\epsilon = 5$ to glean information regarding participation tradeoffs for different $\tau$ values. In Figure \ref{fig:vary-tau-participation} we provide a detailed breakdown of how modulating $\tau$ values can impact the participation reached by the system at ex post PNE. The most salient observation from the plots is that for the given choices of $\tau$ there is not much change in participation. This is due to the fact that for $\tau = 200$ relatively few pools are saturated at representative ex post PNE, hence the relative changes in $\tau$ we explore do not largely change the representative ex post PNE (they still result in few pools being saturated). When delegation is closer to $Cap(f)$, we may see a stronger impact in modulating $\tau$, as larger values of $\tau$ necessarily increase the capacity of all pools, hence providing more leway to allocate delegation in ex post PNE. Figure \ref{fig:vary-tau-decentralization-expenditure} on the other hand shows that our modulations in $\tau$ do not have a large impact on pledge, but they do have a large impact on expenditure. This once again boils down to the number of saturated pools at representative ex post PNE. Though there isn't much of a relative difference in number of pools that are saturated (having a lower impact on decentralization), expenditure is more sensitive to number of pools saturated and hence we see a larger amount of pool rewards being given at representative ex post PNE.

\begin{table}
\centering
\begin{tabular}{ |c|c|c|c|c|c| } 
 \hline
 $\tau$ & 100 & 150 & 200 & 250  \\ 
 \hline
 Ex post SPO stable draws & 499 & 496 & 495 & 496 \\ 
 \hline
\end{tabular}
\caption{The number of ex post SPO stable draws (out of 500) for different $\tau$ values.}
\label{table:modulate-tau}
\end{table}

\begin{figure}
    \centering
    \includegraphics[width=0.45\textwidth]{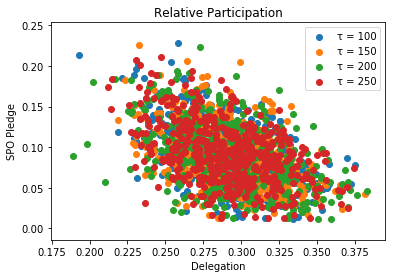}
    \includegraphics[width=0.48\textwidth]{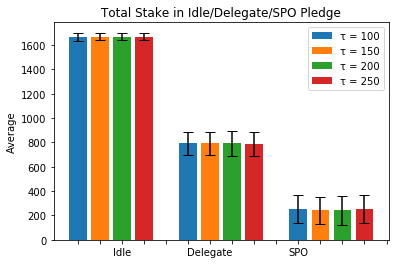}
    \caption{This Figure provides a breakdown of participation for $\tau \in \{100,150,200,250\}$. $\tau$ values correspond to different colors and each point in the plot corresponds to draws of types that gave rise to ex post SPO stability. The axes represent the relative proportion of stake that is used for delegation and SPO pledges. The right bar chart provides average values of absolute stake used by agents being idle, delegators or SPOs respectively for different threshold values.}
\label{fig:vary-tau-participation}
\end{figure}

\begin{figure}
    \centering
    \includegraphics[width=0.45\textwidth]{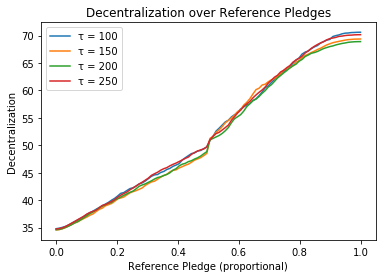}
    \includegraphics[width=0.475\textwidth]{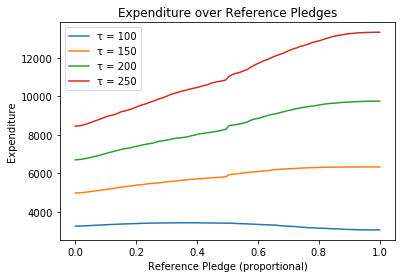}
    \includegraphics[width=0.475\textwidth]{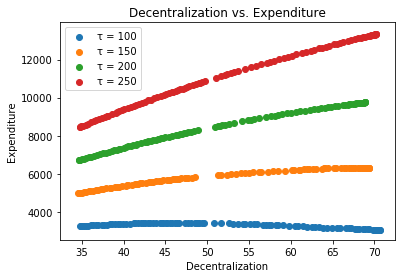}
    \caption{The top two plots provide insight regarding the spread of values the decentralization and expenditure objectives can take for ex post PNE when $\tau \in \{100,150,200,250\}$. The $x$-axis for both of these plots corresponds to different representative PNE as per Algorithm \ref{alg:greedy-delegation}, in which the defining characteristic of a representative PNE is the reference pledge $\bar{\lambda}_j$, which is a proportional value relative to the spread of SPO pledges. The bottom graph simultaneously plots the performance of each representative ex post PNE in terms of decentralization and expenditure.}
\label{fig:vary-tau-decentralization-expenditure}
\end{figure}

\paragraph*{Impact of SPO Threshold in $f^t_{\theta}$}

We modulate the threshold for SPO operation in the ex ante strategy $f^t_\theta$. We consider values $\theta \in \{10,20,30,40,50,60\}$ and Table \ref{table:modulate-threshold} shows the number of ex post SPO stable draws for each given threshold value. The first observation we can make is that the empirical probability that $f^t_\theta$ be ex post SPO stable is decreasing in $\theta$. This makes sense for two reasons; first of all, as $\theta$ increases, pivotal delegates become larger, which in turn increases $r$, the per-unit delegator rewards, thus leaving less rewards for SPOs, and hence decreasing their pool capacity. Second of all, an increased threshold also means that there is more delegation to go around, both from "large" delegates who lie just under the threshold, but also from agents who may have been idle, but with an increased $r$ decide to delegate. All these factors contribute to decreased empirical probability of being ex post SPO stable. Figure \ref{fig:vary-thresh-participation} also provides us a more fine-grained perspective on how participation (and hence $O^P$) changes as a function of $\theta$, where we see once more that increased thresholds decrease SPO operation and increase overall delegation.

\begin{table}
\centering
\begin{tabular}{ |c|c|c|c|c|c|c| } 
 \hline
 $\theta$ & 10 & 20 & 30 & 40 & 50 & 60 \\ 
 \hline
 Ex post SPO stable draws & 500 & 500 & 496 & 478 & 428 & 344 \\ 
 \hline
\end{tabular}
\caption{The number of ex post SPO stable draws (out of 500) for each threshold value of $\theta$.}
\label{table:modulate-threshold}
\end{table}

\begin{figure}
    \centering
    \includegraphics[width=0.45\textwidth]{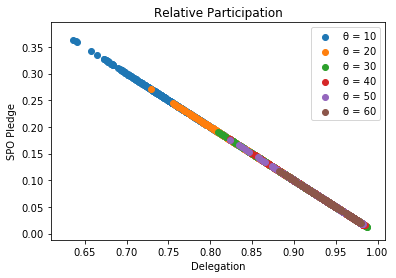}
    \includegraphics[width=0.48\textwidth]{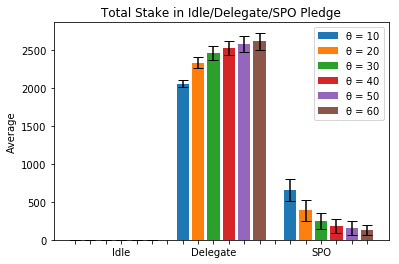}
    \caption{This Figure provides a breakdown of participation as thresholds vary from $\theta \in \{10,20,30,40,50,60\}$. $\theta$ values correspond to different colors and each point in the plot corresponds to draws of types that gave rise to ex post SPO stability. The axes represent the relative proportion of stake that is used for delegation and SPO pledges. As we can see, all points lie on a line indicative of the fact that for no draw do we see idle agents. The right bar chart provides average values of absolute stake used by agents being idle, delegators or SPOs respectively for different threshold values.}
    \label{fig:vary-thresh-participation}
\end{figure}

To gain insight with respect to how decentralization and expenditure are affected by $\theta$, we turn to Figure \ref{fig:vary-thresh-decentralization-expenditure}. The first two images in the figure plot the decentralization and expenditure objectives respectively, as we consider representative PNE of larger reference pledges. Interestingly, we see that as $\theta$ increases, decentralization and expenditure in general increase, and moreover they become more constant as a function of representative ex post PNE reference pledge.  Further observing the third image in the figure, we see that the performance of the $\theta = 10$ threshold is better than others, but we recall that all these points represent ex post PNE, hence depending on the threshold exhibited by players in an ex post PNE, the system can exhibit a multitude of decentralization and expenditure objective values (along all $\theta$
values).

\begin{figure}
    \centering
    \includegraphics[width=0.45\textwidth]{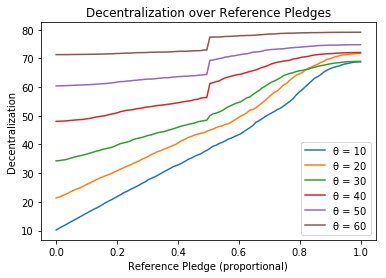}
    \includegraphics[width=0.475\textwidth]{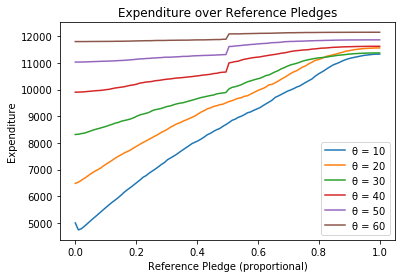}
    \includegraphics[width=0.475\textwidth]{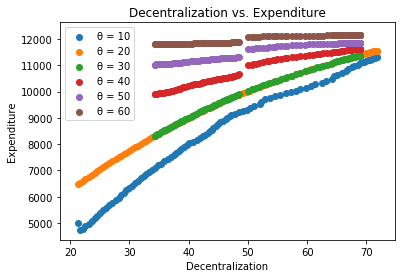}
    \caption{The top two plots provide insight regarding the spread of values the decentralization and expenditure objectives can take for ex post PNE as thresholds vary from $\theta \in \{10,20,30,40,50,60\}$. The $x$ axis for both of these plots correspond to different representative PNE as per Algorithm \ref{alg:greedy-delegation}, in which the defining characteristic of a representative PNE is the reference pledge $\bar{\lambda}_j$, which is a proportional value relative to the spread of SPO pledges. The bottom graph simultaneously plots the performance of each representative ex post PNE in terms of decentralization and expenditure.}
    \label{fig:vary-thresh-decentralization-expenditure}
\end{figure}

\paragraph*{Impact of Inequality of Pareto Distribution}

In this section we modulate $\gamma$ from the Pareto distribution: $\gamma \in \{1.4,1.45,1.5,1.55,1.6\}$. Table \ref{table:modulate-gamma} shows the number of ex post SPO stable draws for each given threshold value. Unlike when we modulate thresholds, we see that changes in $\gamma$ within the range we explored did not have a significant impact on the empirical probability of being ex post SPO stable. 

\begin{table}
\centering
\begin{tabular}{ |c|c|c|c|c|c| } 
 \hline
 $\gamma$ & 1.4 & 1.45 & 1.5 & 1.55 & 1.6 \\ 
 \hline
 Ex post SPO stable draws & 500 & 498 & 496 & 497 & 492 \\ 
 \hline
\end{tabular}
\caption{The number of ex post SPO stable draws (out of 500) for each value of $\gamma$.}
\label{table:modulate-gamma}
\end{table}

We do see qualitatively similar behavior to modulating $\theta$ in terms of participation, decentralization, and expenditure. In terms of participation, Figure \ref{fig:vary-gamma-participation} shows that lower $\gamma$ values result in more {\it stake} participating, but this is simply a reflection of the fact that the resulting Pareto distribution has a heavier tail, and hence the expected stake per player increases, thus increasing the overall stake in the system. The left image from the figure though shows proportional participation, in which we see that proportionally as $\gamma$ increases, there are less SPOs and more delegators. This is also in line with the intuition that larger $\gamma$ values result in distribution with less "high-wealth" individuals, which under threshold strategies are precisely those who become SPOs.

\begin{figure}
    \centering
    \includegraphics[width=0.45\textwidth]{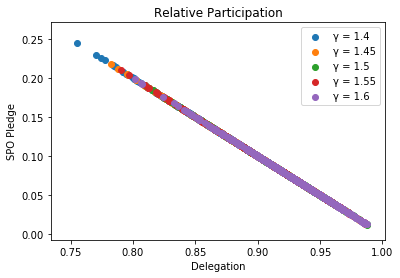}
    \includegraphics[width=0.48\textwidth]{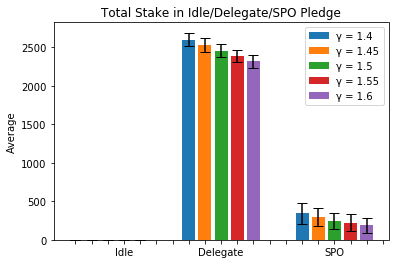}
    \caption{This Figure provides a breakdown of participation as inequality in the Pareto distribution varies from $\gamma \in \{1.4,1.45,1.5,1.55,1.6\}$. $\gamma$ values correspond to different colors and each point in the plot corresponds to draws of types that gave rise to ex post SPO stability. The axes represent the relative proportion of stake that is used for delegation and SPO pledges. As we can see, all points lie on a line indicative of the fact that for no draw do we see idle agents. The right bar chart provides average values of absolute stake used by agents being idle, delegators or SPOs respectively for different threshold values.}
    \label{fig:vary-gamma-participation}
\end{figure}

In Figure \ref{fig:vary-gamma-decentralization-expenditure} we see that $\gamma$ also has an impact on the overall spread of decentralization and expenditure objectives. The range of decentralization and expenditure values is lower than when modulating $\theta$ alone, but we see that $\gamma = 1.6$ results in more decentralization at lower costs. Given the fact that the relative participation breakdown has more delegates for higher $\gamma$ values, this improved performance is most likely from the fact that overall there is less stake in the system in expectation for larger $\gamma$ values, which in turn reduces expenditure and decentralization.

\begin{figure}
    \centering
    \includegraphics[width=0.45\textwidth]{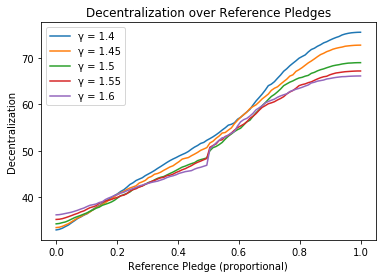}
    \includegraphics[width=0.475\textwidth]{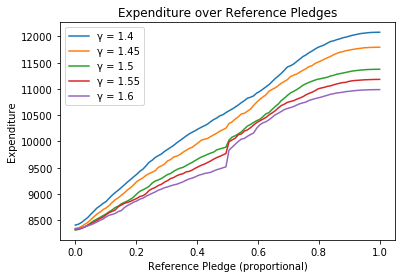}
    \includegraphics[width=0.475\textwidth]{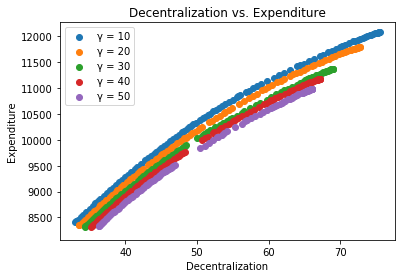}
    \caption{The top two plots provide insight regarding the spread of values the decentralization and expenditure objectives can take for ex post PNE as Pareto inequality varies from $\gamma \in \{1.4,1.45,1.5,1.55,1.6\}$. The $x$ axis for both of thesw plots correspond to different representative PNE as per Algorithm \ref{alg:greedy-delegation}, in which the defining characteristic of a representative PNE is the reference pledge $\bar{\lambda}_j$, which is a proportional value relative to the spread of SPO pledges. The bottom graph simultaneously plots the performance of each representative ex post PNE in terms of decentralization and expenditure.}
    \label{fig:vary-gamma-decentralization-expenditure}
\end{figure}

\paragraph*{Impact of SPO Cost}

We modulate the distribution of SPO costs in two different ways. First we consider settings of $[c_{min},c_{max}]$ that have the same mean of $c = 0.5$ of the baseline parameter settings. In addition to this, we consider $[c_{min},c_{max}]$ settings of a fixed width of $0.1$, but with distinct means. Tables \ref{table:modulate-cwidth} and \ref{table:modulate-cmean} respectively show the empirical frequency of the baseline threshold strategy being ex post SPO stable. The main observation we can draw from the tables is that changes in cos distribution do not have a significant impact for the base parametric setting.

\begin{table}
\centering
\begin{tabular}{ |c|c|c|c| } 
 \hline
 $[c_{min},c_{max}]$ & [0.45,0.55] & [0.4,0.6] & [0.2,0.8] \\ 
 \hline
 Ex post SPO stable draws & 500 & 496 & 500 \\ 
 \hline
\end{tabular}
\caption{The number of ex post SPO stable draws (out of 500) for mean preserving $[c_{min},c_{max}]$ of differing width.}
\label{table:modulate-cwidth}
\end{table}

\begin{table}
\centering
\begin{tabular}{ |c|c|c|c|c|c| } 
 \hline
 $[c_{min},c_{max}]$ & [0.35,0.45] &[0.45,0.55] & [0.55,0.66] & [1.95,2.05] & [4.95,5.05] \\ 
 \hline
 Ex post SPO stable draws & 497 & 500 & 498 & 495 & 496 \\ 
 \hline
\end{tabular}
\caption{The number of ex post SPO stable draws (out of 500) for $[c_{min},c_{max}]$ settings with differing means.}
\label{table:modulate-cmean}
\end{table}

In Figures \ref{fig:vary-cwidth-participation} and \ref{fig:vary-cmean-participation} we see the impact that varying the mean of $[c_{min},c_{max}]$ has on overall participation of the baseline threshold strategy. In addition, Figures \ref{fig:vary-cwidth-decentralization-expenditure} and \ref{fig:vary-cmean-decentralization-expenditure} visualize the changes in decentralization and participation objectives at different representative ex post PNE for different SPO cost settings. We see that increasing SPO costs at this scale do not have much of an effect on decentralization, but they do marginally decrease expenditure. This latter point stems from the fact that larger SPO costs imply that pools have lower capacities, hence they are necessarily earning less pool rewards at saturation due to their smaller sizes. 

\begin{figure} 
    \centering
    \includegraphics[width=0.45\textwidth]{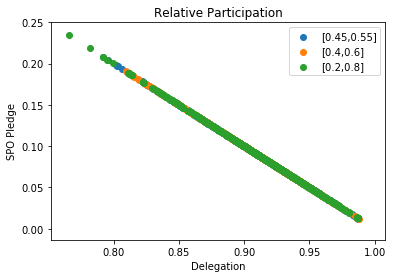}
    \includegraphics[width=0.48\textwidth]{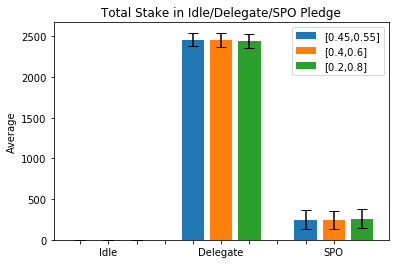}
    \caption{This Figure provides a breakdown of participation as SPO cost distributions vary in width but preserve mean. Different widths correspond to different colors and each point in the plot corresponds to draws of types that gave rise to ex post SPO stability. The axes represent the relative proportion of stake that is used for delegation and SPO pledges. As we can see, all points lie on a line indicative of the fact that for no draw do we see idle agents. The right bar chart provides average values of absolute stake used by agents being idle, delegators or SPOs respectively for different threshold values.}
    \label{fig:vary-cwidth-participation}
\end{figure}

\begin{figure}
    \centering
    \includegraphics[width=0.45\textwidth]{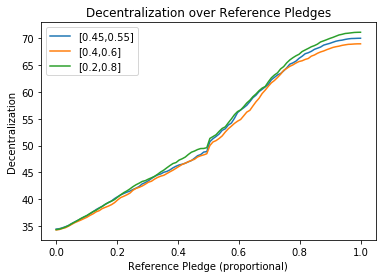}
    \includegraphics[width=0.475\textwidth]{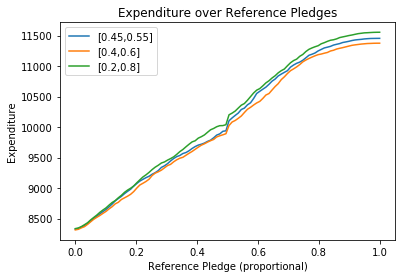}
    \includegraphics[width=0.475\textwidth]{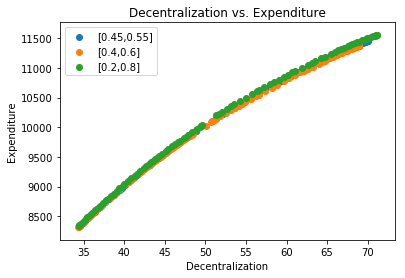}
    \caption{The top two plots provide insight regarding the spread of values the decentralization and expenditure objectives can take for ex post PNE cost distributions vary in width while preserving means. The $x$ axis for both of these plots correspond to different representative PNE as per Algorithm \ref{alg:greedy-delegation}, in which the defining characteristic of a representative PNE is the reference pledge $\bar{\lambda}_j$, which is a proportional value relative to the spread of SPO pledges. The bottom graph simultaneously plots the performance of each representative ex post PNE in terms of decentralization and expenditure.}
    \label{fig:vary-cwidth-decentralization-expenditure}
\end{figure}

\begin{figure} 
    \centering
    \includegraphics[width=0.45\textwidth]{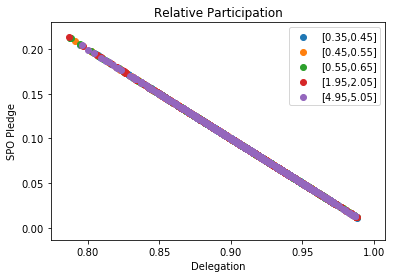}
    \includegraphics[width=0.48\textwidth]{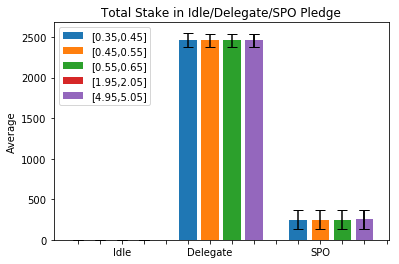}
    \caption{This Figure provides a breakdown of participation as SPO cost distributions vary in mean. Different means correspond to different colors and each point in the plot corresponds to draws of types that gave rise to ex post SPO stability. The axes represent the relative proportion of stake that is used for delegation and SPO pledges. As we can see, all points lie on a line indicative of the fact that for no draw do we see idle agents. The right bar chart provides average values of absolute stake used by agents being idle, delegators or SPOs respectively for different threshold values.}
    \label{fig:vary-cmean-participation}
\end{figure}

\begin{figure}
    \centering
    \includegraphics[width=0.45\textwidth]{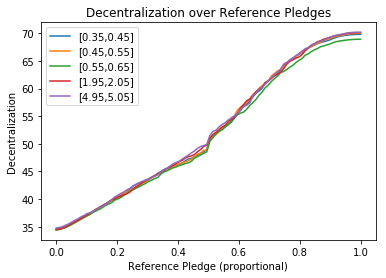}
    \includegraphics[width=0.475\textwidth]{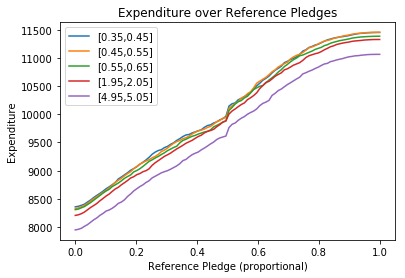}
    \includegraphics[width=0.475\textwidth]{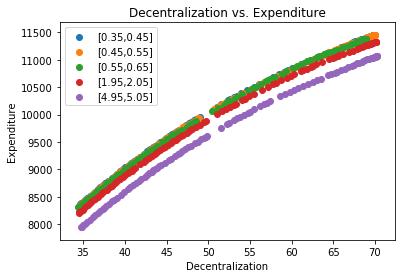}
    \caption{The top two plots provide insight regarding the spread of values the decentralization and expenditure objectives can take for ex post PNE cost distributions vary in mean. The $x$ axis for both of thesw plots correspond to different representative PNE as per Algorithm \ref{alg:greedy-delegation}, in which the defining characteristic of a representative PNE is the reference pledge $\bar{\lambda}_j$, which is a proportional value relative to the spread of SPO pledges. The bottom graph simultaneously plots the performance of each representative ex post PNE in terms of decentralization and expenditure.}
    \label{fig:vary-cmean-decentralization-expenditure}
\end{figure}

\section{Conclusion}
\label{sec:conclusion}

In this work, we have provided a multi-objective framework for studying tradeoffs inherent in delegation systems for PoS cryptocurrencies. We began by providing a broad game theoretic framework for incentives in delegation systems, and successively narrowed down the game at hand to both represent key characteristics of existing PoS delegation systems, and also be tractable to study in a Bayesian framework. We provide key sufficient conditions for equilibria in the one-shot and Bayesian setting and use this characterization to study the potential performance of various payment schemes with respect to three key objectives: participation, decentralization and expenditure. The computational tools we provide give us insight with respect to the inherent tradeoffs system designers may face when attempting to maximize for these three natural objectives. In particular, our experimental results show scenarios in which modulating payment schemes can provide the flexibility needed to prioritize specific objectives amongst the 3, albeit at a potential detriment to the remaining objectives. 

With increased usage of delegation in PoS protocols, it will be important to conceptualize inherent tradeoffs faced by system designers, and techniques such as ours can inform a collective decision in terms of what reward sharing schemes to use depending on overall priorities. We believe our work is a preliminary foray into the tradeoffs that must necessarily be struck in stake delegation systems. Indeed there remain many future directions of work which can further elucidate system tradeoffs. For example, a natural thread would be to relax the constraints inherent in proper delegation games (for example in the $\rho$ functions used), though this would necessitate a much more involved game-theoretic analysis. In addition, we made the simplifying assumption that players either choose to be idle, delegate or be SPOs. In practice, agents can split their stake into many of these roles, and it would be important to see what tradeoffs arise with an increased action space. Finally, as delegation schemes become more prevalent, it may very well be the case that multiple reward schemes interact within a given system, in which case it would be important to understand the potential implications of players being able to choose which delegation schemes to participate in or split their delegated stake across multiple such schemes.

%%
%% Bibliography
%%

%% Please use bibtex, 

\bibliography{refs.bib}

\begin{thebibliography}{10}

\bibitem{azouvi2021levels}
Sarah Azouvi.
\newblock {\em Levels of Decentralization and Trust in Cryptocurrencies: Consensus, Governance and Applications}.
\newblock PhD thesis, UCL, 2021.

\bibitem{brunjes2020reward}
Lars Br{\"u}njes, Aggelos Kiayias, Elias Koutsoupias, and Aikaterini-Panagiota Stouka.
\newblock Reward sharing schemes for stake pools.
\newblock In {\em 2020 IEEE european symposium on security and privacy (EuroS\&p)}, pages 256--275. IEEE, 2020.

\bibitem{csirik1991heuristics}
J{\'a}nos Csirik and Johannes Bartholomeus~Gerardus Frenk.
\newblock Heuristics for the 0-1 min-knapsack problem.
\newblock {\em Acta Cybernetica}, 10(1-2):15--20, 1991.

\bibitem{gencer2018decentralization}
Adem~Efe Gencer, Soumya Basu, Ittay Eyal, Robbert van Renesse, and Emin~Gün Sirer.
\newblock Decentralization in bitcoin and ethereum networks, 2018.
\newblock \href {https://arxiv.org/abs/1801.03998} {\path{arXiv:1801.03998}}.

\bibitem{gersbach2022staking}
Hans Gersbach, Akaki Mamageishvili, and Manvir Schneider.
\newblock Staking pools on blockchains.
\newblock {\em arXiv preprint arXiv:2203.05838}, 2022.

\bibitem{gogol2024sok}
Krzysztof Gogol, Yaron Velner, Benjamin Kraner, and Claudio Tessone.
\newblock Sok: Liquid staking tokens (lsts).
\newblock {\em arXiv preprint arXiv:2404.00644}, 2024.

\bibitem{karakostas2022sok}
Dimitris Karakostas, Aggelos Kiayias, and Christina Ovezik.
\newblock Sok: A stratified approach to blockchain decentralization, 2022.
\newblock \href {https://arxiv.org/abs/2211.01291} {\path{arXiv:2211.01291}}.

\bibitem{LEE2021278}
Jaeseung Lee, Byungheon Lee, Jaeyoung Jung, Hojun Shim, and Hwangnam Kim.
\newblock Dq: Two approaches to measure the degree of decentralization of blockchain.
\newblock {\em ICT Express}, 7(3):278--282, 2021.

\bibitem{Leonardos2019OceanicGC}
Nikos Leonardos, Stefanos Leonardos, and Georgios Piliouras.
\newblock Oceanic games: Centralization risks and incentives in blockchain mining.
\newblock In {\em MaRBLe}, 2019.

\bibitem{ovezik2022decentralization}
Christina Ovezik and Aggelos Kiayias.
\newblock Decentralization analysis of pooling behavior in cardano proof of stake.
\newblock In {\em Proceedings of the Third ACM International Conference on AI in Finance}, pages 18--26, 2022.

\bibitem{pareto1964cours}
Vilfredo Pareto.
\newblock {\em Cours d'{\'e}conomie politique}, volume~1.
\newblock Librairie Droz, 1964.

\bibitem{rocketpool2021}
David Rugendyke.
\newblock Rocket pool: Decentralized ethereum proof of stake network, 2021.
\newblock Accessed: 2024-05-20.
\newblock URL: \url{https://rocketpool.net/files/rocketpool-whitepaper.pdf}.

\bibitem{lido2021}
Vasiliy Shapovalov, Konstantin Lomashuk, Jordan Fish, Will Harborne, and Kasper Rasmussen.
\newblock Lido: A decentralized solution for liquid staking, 2021.
\newblock Accessed: 2024-05-20.
\newblock URL: \url{https://research.lido.fi/t/lido-whitepaper/2}.

\bibitem{lin2021measuring}
Balaji Srinivasan and Leland Lee.
\newblock Quantifying decentralization, 2017.
\newblock URL: \url{https: //news.earn.com/quantifying- decentralization- e39db233c28e Medium}.

\bibitem{decentrfruitchain}
A.~Stouka and T.~Zacharias.
\newblock On the (de)centralization of fruitchains.
\newblock In {\em 2023 2023 IEEE 36th Computer Security Foundations Symposium (CSF) (CSF)}, pages 299--314, Los Alamitos, CA, USA, jul 2023. IEEE Computer Society.
\newblock URL: \url{https://doi.ieeecomputersociety.org/10.1109/CSF57540.2023.00020}, \href {https://doi.org/10.1109/CSF57540.2023.00020} {\path{doi:10.1109/CSF57540.2023.00020}}.

\bibitem{tauhidul2009approximation}
Islam~Mohammad Tauhidul.
\newblock Approximation algorithms for minimum knapsack problem.
\newblock {\em Master’s degree Thesis, university of lethbridge}, 2009.

\bibitem{etherfi2023}
Ether.fi Team.
\newblock Ether.fi: A decentralized, non-custodial liquid staking protocol, 2023.
\newblock Accessed: 2024-05-20.
\newblock URL: \url{https://etherfi.gitbook.io/etherfi/ether.fi-whitepaper}.

\end{thebibliography}

\end{document}